\documentclass[twocolumn]{autart}

\usepackage{amssymb,amsmath,graphicx,epstopdf,color,picins,amsfonts,cite,harvard}

\newtheorem{theorem}{Theorem}[section]
\newtheorem{remark}[theorem]{Remark}
\newtheorem{corollary}[theorem]{Corollary}

\newtheorem{lemma}[theorem]{Lemma}
\newtheorem{proposition}[theorem]{Proposition}
\newtheorem{definition}[theorem]{Definition}

\newenvironment{proof}{\textbf{Proof. }}{\hfill{$\square$}}
\newcommand{\abs}[1]{\lvert#1\rvert}
\begin{document}
	
\begin{frontmatter} 
\title{
Resilient Consensus of Second-Order Agent Networks: \\[.5mm]
Asynchronous Update Rules with Delays}
		
\thanks[footnoteinfo]{This work was supported in part 
by the Japan Science and Technology Agency under the EMS-CREST program
and by JSPS under Grant-in-Aid for Scientific Research Grant No.~15H04020.}

\author[mit]{Seyed Mehran Dibaji}\ead{dibaji@mit.edu}~and~
\author[tit]{Hideaki Ishii}\ead{ishii@c.titech.ac.jp}

\address[mit]{Department of Mechanical Engineering, Massachusetts Institute of Technology, Cambridge, MA 02139, USA}
	
\address[tit]{%
Department of Computer Science,Tokyo Institute of Technology, Yokohama 226-8502, Japan}

\begin{abstract}
We study the problem of resilient consensus of 
sampled-data multi-agent networks with double-integrator dynamics.
The term resilient points to algorithms considering the presence of attacks by faulty/malicious agents in
the network. Each normal agent updates its state based on a
predetermined control law using its neighbors' information 
which may be delayed
while misbehaving agents make updates arbitrarily and might 
threaten the consensus within the network. Assuming that the 
maximum number of malicious agents in the system is known, 
we focus on algorithms where each normal agent ignores large and small position values among its
neighbors to avoid being influenced by malicious agents. The malicious agents are assumed to be omniscient in that they know the updating times and delays and can collude with each other.
We deal with both synchronous 
and {\color{black} partially} asynchronous cases 
with delayed information and derive topological conditions in terms of graph robustness. 
\vspace*{-2mm}
\end{abstract}
\end{frontmatter}

\section{Introduction}

In recent years, much attention has been devoted to the study of networked control systems with an emphasis on cyber security.
Due to communications through shared networks, there are many vulnerabilities
for potential attacks, which can result in irreparable damages. Conventional control approaches are often not applicable for resiliency against such unpredictable 
but probable misbehaviors in networks (e.g., \cite{security_csm:15}). 
One of the most essential problems in networked 
multi-agent systems is consensus where 
agents interact locally to achieve the global goal of reaching a common 
value. Having a wide variety of applications in UAV formations, 
sensor networks, power systems, and so on, consensus problems have been studied extensively 
\cite{mesbahi,renbook2}. 
\textit{Resilient consensus} points to the case where some agents 
in the network anonymously try to mislead the others or are subject 
to failures. Such malicious agents do not comply with the predefined 
interaction rule and might even prevent the normal agents from 
reaching consensus.
This type of problems has a rich history in distributed 
algorithms in the area of computer science (see, e.g., \cite{lynch}) 
where the agents' values are often discrete and finite. 
It is interesting that 
randomization sometimes play a crucial role;
see also \cite{MotRag:95,TemIsh,dibajiishiitempoACC2016}.

In such problems, the non-faulty agents cooperate by interacting locally with each other 
to achieve agreement. There are different techniques to mitigate the effects of attacks.
In some solutions, each agent has a bank of observers to identify the faulty agents within the network using their past information. Such solutions are formulated as a kind of fault detection and isolation problems \cite{pasq,shames,sundaram}. 
However, identifying the malicious agents can be challenging and 
requires much information processing at the agents. 
In particular, these techniques usually necessitate each agent to know 
the topology of the entire network. This global information typically is not desirable in distributed algorithms. 
To overrule the effects of $f$ malicious agents, 
the network has to be at least $(2f + 1)$-connected.

There is another class of algorithms for resilient consensus 
where each normal agent disregards the most 
deviated agents in the updates.
In this case, 
they simply neglect the information received from suspicious 
agents or those with unsafe values whether or not they 
are truly misbehaving. This class of algorithms has been 
extensively used in computer 
science \cite{Azadmanesh2002,Azevedo,zohir,lynch,plunkett,Vaidya} 
as well as control 
\cite{dibajiishiiSCL2015,dibajiishiiNecSys2015,LeBlancPaper,zhang1};
see also \cite{Feng,khanafer} for related problems.
They are often called Mean Subsequence Reduced (MSR) algorithms,
which was coined in \cite{Azadmanesh1993}.
Until recently, this strategy had been studied mostly in the case where the agent networks form complete graphs. The authors of \cite{zhang1} have given a thorough study for the non-complete case and have shown that the traditional connectivity measure is not adequate for MSR-type algorithms to achieve resilient consensus. They then introduced a new 
notion called graph robustness. We note that most of these works have dealt with
single-integrator and synchronous agent networks. 

In this paper, we consider agents having second-order dynamics, which is a common model for
autonomous mobile robots and vehicles. 
Such applications in fact provide motivations different from those in computer science as we will see.
In our previous paper \cite{dibajiishiiSCL2015},
an MSR-type algorithm has been applied to sampled-data second-order agent networks. We have considered the problem of resilient consensus when each agent is affected by at most $f$ malicious agents among its neighbors. Such a model is called $f$-local malicious.
We have established a sufficient condition
on the underlying graph structure to reach consensus.
It is stated in terms of graph robustness
and is consistent with the result in
\cite{zhang1} for the first-order agent case.

Here, the focus of our study is on the so-called $f$-total model,
where the total number of faulty agents is at most $f$,
which has been dealt with in, e.g., \cite{Azadmanesh2002,zohir,Azadmanesh1993,LeBlancPaper,zhang1,lynch,Vaidya}.
We derive a necessary and sufficient condition to achieve
resilient consensus by an MSR-like algorithm. Again,
we show that graph robustness in the network is the relevant
notion. However, the $f$-total model assumes fewer malicious
agents in the system, and hence, the condition will
be shown to be less restrictive than that for the $f$-local case.
The works \cite{zohir,Azadmanesh1993,lynch,Vaidya} have 
studied this model for the first-order agents case, but based on 
the Byzantine malicious agents, which are allowed to send 
different values to their neighbors. 
Such attacks may be impossible, e.g., if the measurements 
are made by on-board sensors in mobile robots. 

Under the $f$-total model,
we solve the resilient consensus problem using MSR-type algorithms
for two different updating rules: Synchronous and partially asynchronous\footnote{%
The term \textit{partially} asynchronous refers to the case where agents share some
level of synchrony by having the same sampling times; however, they make updates at different times
based on delayed information \cite{Bertsekas}. 
This is in contrast to the \textit{fully} asynchronous case where 
agents must be facilitated with their own clocks; such settings are
studied in, e.g., \cite{JiahuHirche}.
}.
In the synchronous case, all agents simultaneously make updates
at each time step using the current information of their neighbors.
By contrast, in the asynchronous case, normal agents may decide to
update only occasionally and moreover, the neighbors' data may
be delayed. This is clearly a more vulnerable situation,
allowing the adversaries to take advantage by quickly moving around.
We consider the worst-case scenarios where the malicious agents
are aware of the updating times and even the delays in the information
of normal agents. The normal agents on the other hand
are unaware of the updating times of their neighbors and hence
cannot predict the plans of adversaries. 
For both cases,
we develop graph robustness conditions for the overall network
topologies. It will be shown that the synchronous updating rules
require less connectivity than the asynchronous counterpart;
see also \cite{dibajiishiiNecSys2015} regarding corresponding 
results for first-order agent systems. 

The main features of this work are three-fold:
(i)~We deal with second-order agents, which
are more suitable for modeling networks of vehicles, but
exhibit more complicated dynamics
in comparison to the single-order case.
(ii)~For the malicious agents, we consider the $f$-total model, which is
less stringent than the $f$-local case, but the analysis is more involved.
(iii)~In the asynchronous case with delayed information,
we introduce a new update scheme,
which is more natural in view of the current research in the area
of multi-agent systems than those 
based on the so-called rounds, 
commonly employed in computer science as we discuss later. 

The paper is organized as follows.  
Section~\ref{sect:problemsetting} presents preliminaries for 
introducing the problem setting.
Section~\ref{sect:SynchSystems} focuses on resilient consensus based on 
synchronous update rules. 
Section~\ref{sect: asynchsystems} is devoted to the problem 
of partial asynchrony with delayed information.
We illustrate the results
through a numerical example in Section~\ref{sect: simulations}. 
Finally, Section~\ref{sect: conclusion} concludes the paper. 
The material of this paper 
appears in \cite{dibajiishiiACC2015,dibajiishiiCDC2015} in 
preliminary forms; here, we present improved results with 
full proofs and more discussions.

\section{Problem Setup}\label{sect:problemsetting}

\subsection{Graph Theory Notions}\label{graphnot}
We recall some concepts on graphs \cite{mesbahi}.
A directed graph (or digraph) with $n$ nodes $(n> 1)$ is defined 
as $\mathcal{G}=(\mathcal{V},\mathcal{E})$ with the node set 
$\mathcal{V}=\{1,\ldots,n\}$ and the edge set 
$\mathcal{E}\subseteq \mathcal{V}\times\mathcal{V}$.
The edge $(j,i)\in \mathcal{E}$ means that node $i$ has access 
to the information of node $j$. 
{\color{black}If $\mathcal{E}=\{(i,j):\,i,j\in \mathcal{V},~i \neq j\}$, the graph is said to be complete.}
For node $i$, the set of its neighbors, 
denoted by $\mathcal{N}_i=\{j:(j,i)\in \mathcal{E}\}$,
consists of all nodes having directed edges toward $i$.
The degree of node $i$ is the number of its neighbors and 
is denoted by ${d}_i=\abs{\mathcal{N}_i}$.  
The adjacency matrix $A=[a_{ij}]$ is given by $a_{ij}\in[\gamma,1)$ 
if $(j,i)\in \mathcal{E}$ and 
otherwise $a_{ij}=0$, where $\gamma > 0$ is a fixed lower bound. We assume that $\sum_{j=1,j\neq i}^{n} a_{ij} \leq 1$.
Let $L=[{l}_{ij}]$ be the Laplacian matrix of $\mathcal{G}$, whose entries 
are defined as $l_{ii}=\sum_{j=1,j\neq i}^{n}a_{ij}$ 
and $l_{ij}=-a_{ij},~i\neq j$;
we can see that the sum of the elements of each row 
of $L$ is zero. 

A path from node $v_1$ to $v_p$ is 
a sequence $(v_1,v_2,\ldots,v_p)$ in which 
$(v_i,v_{i+1}) \in \mathcal{E}$ for $i=1, \ldots,p-1$. 
If there is a path between each pair of nodes, 
the graph is said to be strongly connected. 
A directed graph is said to have a directed spanning tree if
there is a node from which there is a path to every other node 
in the graph. 

For the MSR-type resilient consensus algorithms,
the critical topological notion is graph robustness, 
which is a connectivity measure of graphs. 
Robust graphs were introduced in \cite{zhang1} 
for the analysis of resilient consensus of first-order 
multi-agent systems. 

\begin{definition}\label{(def:robustgraph)}\rm
The digraph $\mathcal{G}$ is $(r,s)$-robust $(r,s<n)$ if 
for every pair of nonempty disjoint subsets 
$\mathcal{S}_1,\mathcal{S}_2 \subset \mathcal{V}$, 
at least one of the following conditions is satisfied:
\[
{\color{black}
1.~\mathcal{X}_{\mathcal{S}_1}^r =\mathcal{S}_1,~~~
2.~\mathcal{X}_{\mathcal{S}_2}^r=\mathcal{S}_2,~~~
3.~\abs{\mathcal{X}_{\mathcal{S}_1}^r} +\abs{\mathcal{X}_{\mathcal{S}_2}^r} \geq s,}
\]
where $\mathcal{X}^r_{\mathcal{S}_{\ell}}$ is the set of all nodes in ${\mathcal{S}_{\ell}}$ which have at least $r$ incoming edges from outside of ${\mathcal{S}_{\ell}}$.
In particular, graphs which are $(r,1)$-robust are called $r$-robust.
\end{definition}

The following lemma helps to have a better understanding 
of $(r,s)$-robust graphs \cite{LeBlancPhD}.

\begin{lemma}\label{lemma:robust graphs}\rm
	For an $(r,s)$-robust graph $\mathcal{G}$, the following hold:
	\begin{enumerate}
		\item[(i)] $\mathcal{G}$ is $(r',s')$-robust, where $0\leq r'\leq r$ and $1 \leq s'\leq s$, and in particular, it is $r$-robust.
		\item[(ii)] $\mathcal{G}$ is $(r-1,s+1)$-robust.
		\item[(iii)] $\mathcal{G}$ is at least $r$-connected, but an $r$-connected graph is not necessarily $r$-robust. 
		\item[(iv)] $\mathcal{G}$ has a directed spanning tree.
		\item[(v)] $r \leq \lceil n/2 \rceil$. Also, if $\mathcal{G}$ is a complete graph, 
		then it is $(r',s)$-robust for all $0<r'\leq \lceil n/2 \rceil$ and $1 \leq s \leq n$.
	\end{enumerate}
	Moreover, a graph is $(r,s)$-robust 
	if it is $(r+s-1)$-robust.
\end{lemma}

 \begin{figure}
	\centering
	\includegraphics[width=24mm]{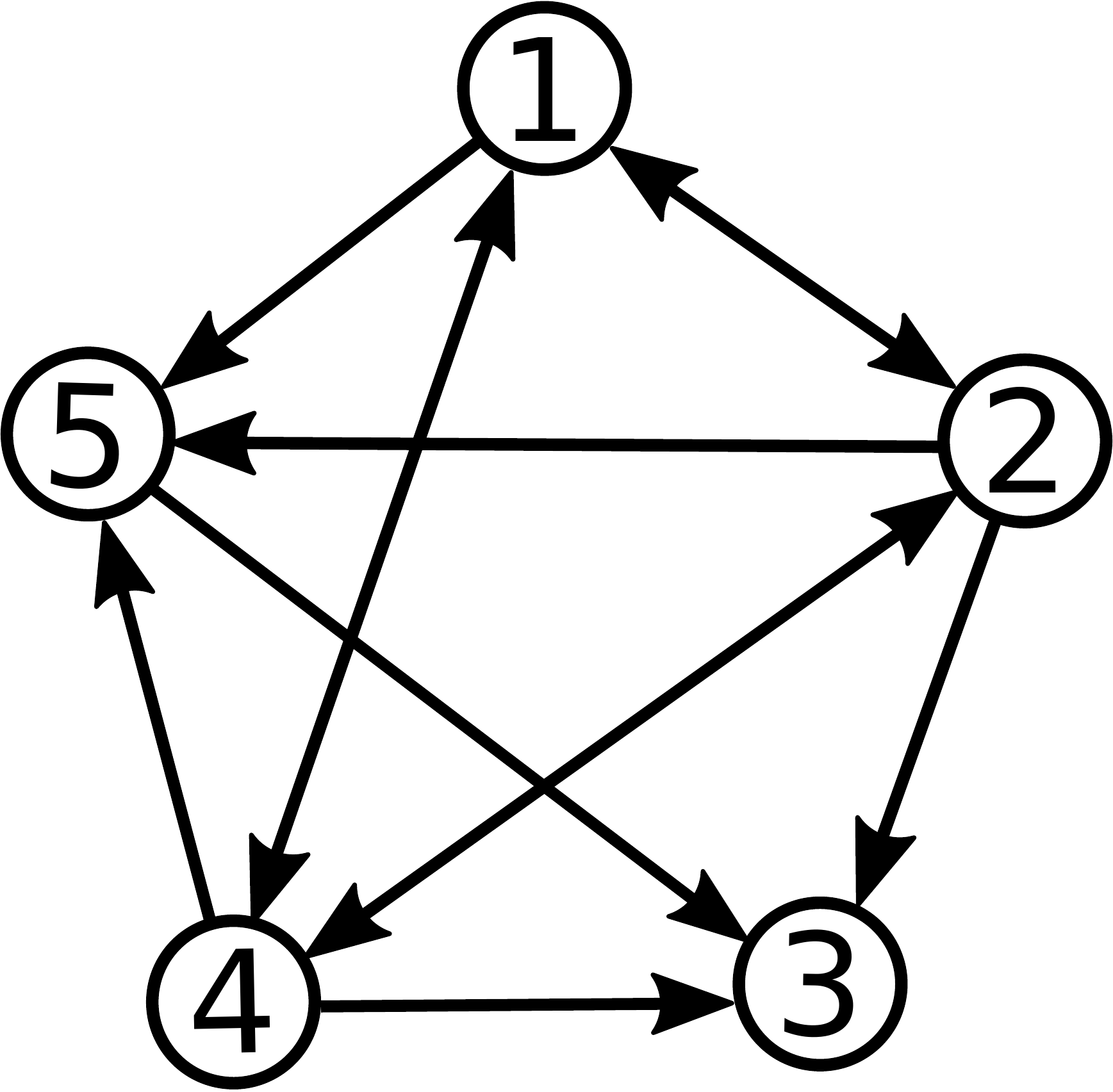}
	\vspace*{-1mm}
	\caption{A graph which is  $(2,2)$-robust but not $3$-robust.}
	\label{fig:robustgraph}
\end{figure}

It is clear that $(r,s)$-robustness is more restrictive than $r$-robustness. 
The graph with five nodes in Fig.~\ref{fig:robustgraph} 
is $(2,2)$-robust, but not $3$-robust;
further, removing any edge destroys its $(2,2)$-robustness.
In general, to determine if a given graph
has a robustness property is computationally difficult 
since the problem involves combinatorial aspects. 
It is known that random graphs become robust
when their size tends to infinity \cite{zhangfata}.

\subsection{Second-Order Consensus Protocol}

Consider a network of agents whose interactions are represented by the directed graph $\mathcal{G}$. 
Each agent $i\in\mathcal{V}$ has a double-integrator dynamics given by
\begin{align*}
\begin{split}
 \dot{x_i}(t)=v_{i}(t),~~ \dot{v}_i(t)=u_i(t),\qquad {}  i=1,\ldots,n, 
\end{split}
\end{align*}
where $x_i(t)\in{\mathbb{R}}$ and $v_i(t)\in{\mathbb{R}}$ are 
its position and velocity, respectively, and $u_i(t)$ is the control input.
We discretize the system 
with sampling period $T$ as
\begin{align}
\begin{split}
x_i[k+1]&=x_i[k]+Tv_i[k]+\frac{T^2}{2}u_i[k],\\
v_i[k+1]&=v_i[k]+Tu_i[k],\qquad  i=1,\ldots,n,
\end{split}
\label{eqn: SecondOrderDynamicsDicrete}
\end{align}
where $x_i[k]$, $v_i[k]$, and $u_i[k]$ are, respectively, the position, 
the velocity, and the control input of agent $i$ at $t=kT$ 
for $k\in \mathbb{Z}_+$. 
{\color{black} Our discretization is based on control inputs 
generated by zeroth order holds; other methods 
are employed in, e.g., \cite{LinP,JiahuHirche}.}

At each time step $k$, the agents update their positions and 
velocities based on the time-varying topology of the graph $\mathcal{G}[k]$, 
which is a subgraph of $\mathcal{G}$ and is specified later. 
In particular, the control uses
the relative positions with its neighbors and its own velocity
\cite{renbook2}:
\begin{equation}
  u_i[k] = -\sum_{j\in\mathcal{N}_i}
              a_{ij}[k][(x_i[k]-\delta_i)-(x_j[k]-\delta_j)]-\alpha v_i[k],
 \label{eq3}
\end{equation}
where $a_{ij}[k]$ is the $(i,j)$ entry of the adjacency matrix $A[k]\in \mathbb{R}^{n\times n}$ 
corresponding to $\mathcal{G}[k]$, $\alpha$ is a positive scalar, and $\delta_i \in \mathbb{R}$ is a constant representing the desired relative position of agent $i$ 
in a formation. 

The agents' objective is consensus 
in the sense that they come to formation 
and then stop asymptotically: 
\[
x_i[k]-x_j[k] \rightarrow \delta_i-\delta_j,~~
v_i[k]\rightarrow 0~\text{as $k\rightarrow\infty$},~~\forall i,j\in \mathcal{V}. 
\]
In \cite{renbook2}, it is shown that if there is some 
$\ell_0\in\mathbb{Z}_+$ such that for any nonnegative integer $k_0$, the union of $\mathcal{G}[k]$ across $k \in [k_0, k_0+\ell_0]$ has a directed spanning tree, 
then consensus can be obtained 
under the control law \eqref{eq3} by properly choosing 
$\alpha$ and $T$.

In this paper, we study the case where some agents 
malfunction due to failure, disturbances, or attacks.
In such circumstances, 
they may not follow the predefined update rule \eqref{eq3}. 
In the next subsection, we introduce necessary definitions and 
then formulate the resilient 
consensus problem in the presence of malicious agents. 

Finally, we represent the agent system 
in a vector form. 
Let $\hat{x}_i[k]=x_i[k]-\delta_i$, 
$\hat{x}[k]=\left[\hat{x}_1[k]\,\cdots\,\hat{x}_n[k]\right]^T$,
and $v[k]=\left[v_1[k]\,\cdots\,v_n[k]\right]^T$. 
{\color{black}For the sake of simplicity, hereafter, the agents' 
positions refer to $\hat{x}[k]$ and not $x[k]$.}
The system \eqref{eqn: SecondOrderDynamicsDicrete} 
then becomes
\begin{align}
\begin{split}
  \hat{x}[k+1] &= \hat{x}[k] + T v[k] + \frac{T^2}{2} u[k],\\
  v[k+1]  &= v[k] + T u[k], 
\end{split}
\label{eqn:rv}
\end{align}
and the control law \eqref{eq3} can be written as
\begin{equation}
  u[k] = - L[k] \hat{x}[k] - \alpha v[k],
\label{eqn:u}
\end{equation}
where $L[k]$ is the Laplacian matrix for the graph $\mathcal{G}[k]$.

\subsection{Resilient Consensus}

We introduce notions related to malicious agents and consensus in the presence of such agents \cite{zhang1,lynch,Vaidya}.

\begin{definition}\label{normal node}\rm
Agent $i$ is called normal if it updates its state 
based on the predefined control \eqref{eq3}. 
Otherwise, it is called malicious
and may make arbitrary updates. The index set of malicious agents
is denoted by $\mathcal{M}\subset\mathcal{V}$. 
The numbers of normal agents and malicious agents are denoted by $n_N$ and $n_M$, respectively.
\end{definition}

We assume that an upper bound is available 
for the number of misbehaving agents in the entire network or at least in each normal agent's neighborhood. 

\begin{definition}\rm
The network is $f$-total malicious if the number $n_M$ of faulty agents 
is at most $f$, i.e., $n_M \leq f$.
%
%
On the other hand,
the network is $f$-local malicious if the number of malicious agents in the neighborhood of each normal agent $i$ is bounded by $f$, i.e., $\abs{\mathcal{N}_i \cap \mathcal{M}} \leq f$. 
\end{definition} 


According to the model of malicious agents considered, 
the difference between normal agents and malicious agents 
lies in their control inputs $u_i$: For the normal agents, it is given by \eqref{eq3} 
while for the malicious agents, it is arbitrary. 
On the other hand, the position and velocity dynamics for all 
agents remain the same as \eqref{eqn: SecondOrderDynamicsDicrete}.

We introduce the notion of resilient consensus 
for the network of second-order agents 
\cite{dibajiishiiSCL2015}.

\begin{definition}\label{resilient consensus}\rm
If for any possible set of malicious agents, 
any initial positions and velocities, 
and any malicious inputs, 
the following conditions are met, then the network is said to reach resilient consensus:
\begin{enumerate}
\item[1.] Safety:~There exists a bounded interval $\mathcal{S}$ determined by the initial positions and velocities of the normal agents such that 
$\hat{x}_i[k]\in \mathcal{S}$, $i \in \mathcal{V} \backslash \mathcal{M},
k \in \mathbb{Z}_+$. The set $\mathcal{S}$ is called the safety interval.
\item[2.] Agreement:~For some $c\in \mathcal{S}$, it holds that 
$\lim_{k \rightarrow \infty }\hat{x}_i[k]=c$ and 
$\lim_{k \rightarrow \infty } v_i[k]=0$, 
$i \in \mathcal{V} \backslash \mathcal{M}$.
\end{enumerate}
\end{definition}

A few remarks are in order regarding the safety interval 
$\mathcal{S}$.
(i)~The malicious agents may or may not be in $\mathcal{S}$, 
while the normal agents must stay inside
though they may still be influenced by the malicious agents 
staying in $\mathcal{S}$. 
(ii)~We impose the safety condition to ensure that the behavior
of the normal agents remains close to that when no malicious agent 
is present.
(iii)~We do not have a safety interval for velocity 
of normal agents and hence they may even move faster 
than their initial speeds. 


\section{Synchronous Networks}\label{sect:SynchSystems}
\subsection{DP-MSR Algorithm}\label{dp-msr}

We first outline the algorithm 
for achieving consensus in the presence of misbehaving agents
in the synchronous case, where all 
agents make updates 
at every time step. 
The algorithm is called DP-MSR, which stands for 
Double-Integrator Position-Based Mean Subsequence Reduced 
algorithm. 
It was proposed in \cite{dibajiishiiSCL2015} for the $f$-local malicious model.

The algorithm has three steps as follows:
\begin{enumerate}
\item[1.] At each time step $k$, each normal agent $i$ receives 
  the {\color{black}relative position values 
{\color{black}$\hat{x}_j[k]-\hat{x}_i[k]$ of its neighbors
  $j \in \mathcal{N}_i[k]$} and sorts them 
in a decreasing order.}
\item[2.] If there are less than $f$ agents whose 
relative position values are greater than or equal to zero, 
then the normal agent $i$ ignores 
the incoming edges from those agents. Otherwise, it ignores the incoming 
edges from $f$ agents counting from those having the largest relative position
values. 
Similarly, if there are less than $f$ agents whose values are smaller than or equal
to zero, then agent $i$ ignores the incoming edges from those agents. 
Otherwise, it ignores the $f$ incoming edges counting from those having the
smallest relative position values.
\item[3.] Apply the control input \eqref{eq3} by substituting $a_{ij} [k]=0$ for edges $(j,i)$ which are ignored in step 2.
\end{enumerate}

The main feature of this algorithm lies in its simplicity. 
Each normal agent ignores the information received from its neighbors 
which may be misleading. In particular, it ignores
up to $f$ edges from neighbors whose positions are large, and $f$ edges from neighbors whose positions are small. 
The underlying graph $\mathcal{G}[k]$ at time $k$ is determined by
the remaining edges. 
The adjacency matrix $A[k]$ and the Laplacian matrix $L[k]$ 
are determined accordingly. 

The problem for the synchronous agent network can be stated as follows:
Under the $f$-total malicious model, find a condition on the network
topology such that the normal agents
reach resilient consensus based on the DP-MSR algorithm.

\subsection{Matrix Representation}

We provide a modified system model 
when malicious agents are present. To simplify the notation, 
the agents' indices are reordered.
Let the normal
agents take indices $1,\ldots,n_N$ and let the malicious agents 
be $n_N+1,\ldots,n$.
Thus, the vectors representing the positions, velocities, 
and control inputs of all agents consist of two parts as
\begin{equation}
\hat{x}[k]
= \begin{bmatrix}
\hat{x}^{N}[k]\\
\hat{x}^{M}[k]
\end{bmatrix},~~
v[k]
= \begin{bmatrix}
v^{N}[k]\\
v^{M}[k]
\end{bmatrix},~~
u[k]
= \begin{bmatrix}
u^{N}[k]\\
u^{M}[k]
\end{bmatrix},
\label{eqn:partition}
\end{equation}
where the superscript $N$ stands for normal and $M$ for malicious. 
Regarding the control inputs $u^{N}[k]$ and $u^{M}[k]$, 
the normal agents follow \eqref{eq3} while the malicious agents
may not. Hence, they can be expressed as
\begin{equation}\label{normalmalicious control}
\begin{split}
u^N[k] &= -L^N[k] \hat{x}[k]
-\alpha \begin{bmatrix}
I_{n_{N}} & 0
\end{bmatrix}v[k],\\
u^M[k] &: \text{arbitrary},
\end{split} 
\end{equation}
where 
$L^N[k]\in \mathbb{R}^{n_N\times n}$ is the matrix formed by the first $n_N$ rows 
of $L[k]$ associated with normal agents.
The row sums of this matrix $L^N[k]$ are zero as in $L[k]$.

With the control inputs of \eqref{normalmalicious control},
we obtain the model for the overall system \eqref{eqn:rv} as
\begin{align}
\nonumber
 \hat{x}[k+1]
   &=\left(
        I_n-\frac{T^2}{2} 
        \begin{bmatrix}
           L^N[k] \\
           0
        \end{bmatrix}
     \right)\hat{x}[k]+Q v[k] \nonumber\\ 
   &~~\qquad~\qquad\qquad\qquad\mbox{}
     +\frac{T^2}{2}
        \begin{bmatrix}
          0\\
          I_{n_M}
        \end{bmatrix} u^M[k],\label{eqn: PosVel}\\
 v[k+1]
   &= -T \begin{bmatrix}
           L^N[k] \\ 
           0
         \end{bmatrix}
      \hat{x}[k] + R v[k] 
      + T\begin{bmatrix}
           0\\
           I_{n_M}
        \end{bmatrix} u^M[k],
 \nonumber
\end{align}
where the partitioning in the matrices 
is in accordance with the vectors in \eqref{eqn:partition}, 
and $Q$ and $R$ are given by
\begin{equation}\label{eqn: MatrixNotation}
Q=TI_n-\frac{\alpha T^2}{2}\begin{bmatrix}
I_{n_N} & 0 \\ 0 & 0
\end{bmatrix},
~~R=
I_n-\alpha T \begin{bmatrix}
I_{n_N} & 0 \\ 0 & 0
\end{bmatrix}.
\end{equation}
For the sampling period $T$ and the parameter $\alpha$, 
we assume\footnote{%
The condition \eqref{eqn:alpha_T} on $T$ and $\alpha$ 
ensures that the matrix $\bigl[\Phi_{1k}~\Phi_{2k}\bigr]$ 
possesses the properties stated in Lemma~\ref{lemmapositionupdates}. 
We may relax it, for example, by not imposing 
$\sum_{j=1}^{n}a_{ij}[k]$ to be less than 1.
While \eqref{eqn:alpha_T} can be fulfilled by any $T$, the control law 
\eqref{eq3} may make the agents exhibit {\color{black} undesired} oscillatory movements. 
This type of property is also seen in previous works 
on consensus of agents with second-order dynamics; see, e.g.,
\cite{renjournal,renbook2}. We remark that
the condition \eqref{eqn:alpha_T} is less restrictive than that in \cite{JiahuQin2012}. 
{\color{black} In general, shortcomings due to this assumption
should be further studied in future research.} 
}
	\begin{equation}
	1+\frac{T^2}{2} \leq  \alpha T \leq 2- \frac{T^2}{2}.
	\label{eqn:alpha_T}
	\end{equation} 
The following lemma from \cite{dibajiishiiSCL2015} plays 
a key role in the analysis.

\begin{lemma}\label{lemmapositionupdates}\rm
	Under the control inputs \eqref{normalmalicious control}, 
	the position vector $\hat{x}[k]$ of the agents for $k\geq 1$ can be expressed as 
	\begin{align*}
	\hat{x}[k+1] 
	&= \begin{bmatrix}
	\Phi_{1k} & \Phi_{2k}
	\end{bmatrix}
	\begin{bmatrix}  
	\hat{x}[k]\\
	\hat{x}[k-1] 
	\end{bmatrix}\\
	&\hspace*{1.5cm}\mbox{}
	+ \frac{T^2}{2} 
	\begin{bmatrix}
	0 \\
	I_{n_M}
	\end{bmatrix}
	\bigl(
	u^M[k] + u^M[k-1]
	\bigr),
	\end{align*}
	where
	\begin{align*}
	\Phi_{1k}
	&= R+I_n- \frac{T^2}{2}\begin{bmatrix}
	L^N[k]\\  0
	\end{bmatrix},\\
	\Phi_{2k}
	&=  -R -\frac{T^2}{2} \begin{bmatrix} L^N[k-1]\\0 \end{bmatrix}.
	\end{align*}
	Moreover{\color{black}, under \eqref{eqn:alpha_T}},
	the matrix $\bigl[\Phi_{1k}~~\Phi_{2k}\bigr]$ 
	is nonnegative, and 
	the sum of each of its first ${n_N}$ rows is one.
\end{lemma}

\begin{remark}\label{remark1}\rm
It is clear from the lemma that
the controls $u^M[k]$ and $u^M[k-1]$ do not directly enter the new
positions $\hat{x}^N[k+1]$ of the normal agents. 
Moreover, 
the positions of the normal agents 
$\hat{x}^N[k+1]$ for $k\geq 1$
are obtained via the convex combination 
of the current positions $\hat{x}[k]$ and
those from the previous time step $\hat{x}[k-1]$. 
\end{remark}

\subsection{A Necessary and Sufficient Condition}

We are now ready to state the main result for the synchronous case.
Let the interval $\mathcal{S}$ be given by 
\begin{align}\label{eqn:  SafetyInterval}
\mathcal{S}
  &= \biggl[
       \min \hat{x}^N[0]
       + \min \biggl\{0,
                \biggl(
                  T - \frac{\alpha T^2}{2}
                \biggr) v^N[0]
              \biggr\},\notag\\
 &\hspace*{.6cm}
     \max \hat{x}^N[0]
       + \max \biggl\{0,
               \biggl(
                 T - \frac{\alpha T^2}{2}
               \biggr) v^N[0]
              \biggr\}
     \biggr],
\end{align}
where the minimum and the maximum are taken 
over all entries of the vectors.
Note that the interval is determined only by
the initial states
of the normal agents.

\begin{theorem}\label{NecSufSynch}\rm
Under the $f$-total malicious model, the network of agents 
with second-order dynamics using the control in 
\eqref{normalmalicious control} and the DP-MSR algorithm reaches resilient consensus 
if and only if the underlying graph is $(f+1,f+1)$-robust.  The safety interval is given by \eqref{eqn:  SafetyInterval}.
\end{theorem}

\begin{proof}(Necessity)~We prove by contradiction. 
Suppose that the network is not $(f+1,f+1)$-robust. 
Then, there are nonempty disjoint sets $\mathcal{V}_1 , \mathcal{V}_2 \subset \mathcal{V}$ such that none of the conditions 1--3 in Definition~\ref{(def:robustgraph)} holds. Suppose all agents in $\mathcal{V}_1$ have initial positions at $a$ and all agents in $\mathcal{V}_2$ have 
initial positions at $b$ with $a < b$. 
Let all other agents have initial positions taken from the interval 
$(a,b)$ and every agent has 0 as initial velocity. 
From condition~3, we have that $\abs{\mathcal{X}_{\mathcal{V}_1}^{f+1}}+\abs{\mathcal{X}_{\mathcal{V}_2}^{f+1}}\leq f$. Suppose that all agents in $\mathcal{X}_{\mathcal{V}_1}^{f+1}$ and $\mathcal{X}_{\mathcal{V}_2}^{f+1}$ are malicious and keep their values constant. There is at least one normal agent in  $\mathcal{V}_1$ and one normal agent in $\mathcal{V}_2$ by $\abs{\mathcal{X}_{\mathcal{V}_1}^{f+1}} < \abs{\mathcal{V}_1}$ and $\abs{\mathcal{X}_{\mathcal{V}_2}^{f+1}} < \abs{\mathcal{V}_2}$ because conditions 1 and 2 do not hold. These normal agents have $f$ or fewer neighbors outside of their own sets because they are not in ${\mathcal{X}_{\mathcal{V}_1}^{f+1}}$ or ${\mathcal{X}_{\mathcal{V}_2}^{f+1}}$. 
As a result, all normal agents in $\mathcal{V}_1$ and $\mathcal{V}_2$ update based only on the values inside $\mathcal{V}_1$ and $\mathcal{V}_2$ by removing the values received from outside of their sets. This makes their positions at $a$ and $b$ unchanged. 
Hence, 
there will be no agreement among the normal agents.
	
(Sufficiency)~
We first establish the safety condition with 
$\mathcal{S}$ given in \eqref{eqn:  SafetyInterval}, 
i.e., $\hat{x}_i[k]\in\mathcal{S}$ for all $k$ and $i\in\mathcal{V}\setminus\mathcal{M}$.
For $k=0$, it is obvious that the condition holds.
For $k=1$, the positions of normal agents are given 
by \eqref{eqn: PosVel} as
\begin{equation}\label{eqn: k=1 update}
  \hat{x}^N[1]
    = \biggl(
       \begin{bmatrix}
         I_{n_N} & 0
       \end{bmatrix} - \frac{T^2}{2}L^N [0]
      \biggr) \hat{x}[0] 
       + \biggl(
            T-\frac{\alpha T^2}{2}
         \biggr) v^N[0].	
\end{equation}
While the initial velocities of the malicious agents do not 
appear in $\hat{x}^{N}$ at this time, their initial positions 
may have influences. 
However, for a normal agent, if some of its neighbors are 
malicious and are outside the interval $[\min \hat{x}^N[0],\max \hat{x}^N[0]]$,
then they will be ignored by step~2 in DP-MSR
because there are at most $f$ such agents. 
The matrix $[I_{n_N} 0]- (T^2/2)L^N[0]$ in 
\eqref{eqn: k=1 update} is
nonnegative and its row sums are one
because $L^N[0]$ consists of the first $n_N$ rows 
of the Laplacian $L[0]$.
As a result, the first term on the right-hand side of 
\eqref{eqn: k=1 update} is a vector whose
entries are convex combinations of values within 
$[\min \hat{x}^N[0],\max \hat{x}^N[0]]$.
Thus, for each normal agent $i\in\mathcal{V}\setminus\mathcal{M}$, 
it holds that $\hat{x}_i[1]\in\mathcal{S}$. 

Next, to further analyze the normal agents, let
\begin{equation}\label{eqn: SecondOrderVariables}
\begin{split}
  \overline{x}[k]=&\max \left(\hat{x}^N[k],\hat{x}^N[k-1]\right),\\
  \underline{x} [k]=&\min\left(\hat{x}^N[k],\hat{x}^N[k-1]\right).
\end{split}
\end{equation} 
In what follows, we show that $\overline{x}[k]$ 
is a nonincreasing function of $k \geq 1$. 
For $k \geq 2$, 
Lemma~\ref{lemmapositionupdates} and Remark~\ref{remark1}
indicate that the positions of normal agents are
convex combinations of those of 
its neighbors from time $k-1$ and $k-2$. 
If any neighbors of the normal agents
at those time steps are malicious and 
are outside the range of the normal agents' position values, 
they are ignored by step~2 in DP-MSR.
Hence, we have
$\max \hat{x}^N[k] \leq \max \left ( \hat{x}^N[k-1],\hat{x}^N[k-2]\right)$. 
It also easily follows that
$\max \hat{x}^N[k-1] \leq \max \bigl( \hat{x}^N[k-1]$, $\hat{x}^N[k-2]\bigr)$.
Thus, we arrive at
\begin{align*}
\overline{x}[k]&=\max \left ( \hat{x}^N[k],\hat{x}^N[k-1]\right) \\ &\leq \max \left ( \hat{x}^N[k-1],\hat{x}^N[k-2]\right)=\overline{x}[k-1].
\end{align*}
Similarly, $\underline{x}[k]$ is a nondecreasing function of time. 
Thus, we have that for $k\geq 2$, normal agents satisfy
$\hat{x}_i[k]\in\mathcal{S}$, $i\in\mathcal{V}\setminus\mathcal{M}$.
The safety condition has now been proven.

It remains to establish the agreement condition. We start 
with the proof of agreement in the position values. Because $\overline{x}[k]$ and $\underline{x}[k]$ 
are bounded and monotone, their limits exist, which
are denoted by $\overline{x}^{\star}$ and $\underline{x}^{\star}$, 
respectively. 
If 
$\overline{x}^{\star}=\underline{x}^{\star}$, 
then resilient consensus follows. We prove by contradiction, 
and thus assume $\overline{x}^{\star}>\underline{x}^{\star}$. 

Denote by $\beta$ the minimum nonzero entry of 
the matrix $\bigl[\Phi_{1k}~\Phi_{2k}\bigr]$ over all $k$.
This matrix is determined by the structure of the 
graph $\mathcal{G}[k]$ and thus can vary over a finite number
of candidates;
by \eqref{eqn:alpha_T}, we have $\beta\in(0,1)$.
Let $\epsilon _0>0$ and $\epsilon>0$ be sufficiently small that
\begin{align}
& \underline{x}^\star +\epsilon _0
< \overline{x}^\star-\epsilon _0,~~
\epsilon
< \frac{\beta^{n_N}\epsilon_0}{1-\beta ^{n_N}}.
\label{eqn: xstar}
\end{align}
We introduce the sequence $\{\epsilon_\ell\}$ defined by
\[
\epsilon_{\ell+1}
= \beta \epsilon_{\ell} - (1-\beta)\epsilon,~~\ell=0,1,\ldots,n_N-1.
\]
It is easy to see that
$\epsilon_{\ell+1} < \epsilon_{\ell}$ for $i=0,1,\ldots,n_N-1$.
Moreover, by \eqref{eqn: xstar}, they are positive because
\[\epsilon_{n_N}  
= \beta^{n_N} \epsilon_{0}  
- \sum_{\ell=0}^{n_N-1}
\beta^\ell (1-\beta) \epsilon
= \beta^{n_N} \epsilon_{0}  
-\bigl(
1- \beta^{n_N} 
\bigr)\epsilon
> 0.
\]
We also take $k_\epsilon\in\mathbb{Z}_{+}$ such that for $k\geq k_\epsilon$,
it holds that 
$\overline{x}[k] 
< \overline{x}^\star +\epsilon$,
$\underline{x}[k]
> \underline{x}^\star -\epsilon$.
Due to convergence of $\overline{x}[k]$ and $\underline{x}[k]$,
such $k_\epsilon$ exists. 
For the sequence $\{\epsilon_\ell\}$,
let 
\begin{equation}\label{eqn: Thedelayedmaxminsets}
\begin{split}
  \mathcal{X}_{1}(k_{\epsilon}+\ell,\epsilon _\ell)
    &=\{j \in \mathcal{V}:~
         \hat{x}_j[k_{\epsilon}+\ell]>\overline{x}^\star-\epsilon_\ell \},\\ 
       \mathcal{X}_{2}(k_{\epsilon}+\ell,\epsilon _\ell)
    &=\{j \in \mathcal{V}:~ 
             \hat{x}_j[k_{\epsilon}+\ell]
               < \underline{x}^\star+\epsilon_\ell \}. 
\end{split}
\end{equation}
For each fixed $\ell$, these sets are disjoint by \eqref{eqn: xstar} and 
$\epsilon_{{\ell}+1} < \epsilon_{{\ell}}$.
Here, we claim that in a finite number of steps, 
one of these sets will contain no normal agent. 
Note that this contradicts the assumption of
$\overline{x}^\star$ and $\underline{x}^\star$ being limits.

\begin{figure}
	\centering
	\vspace*{-19mm}
	\def \svgwidth{8cm}
	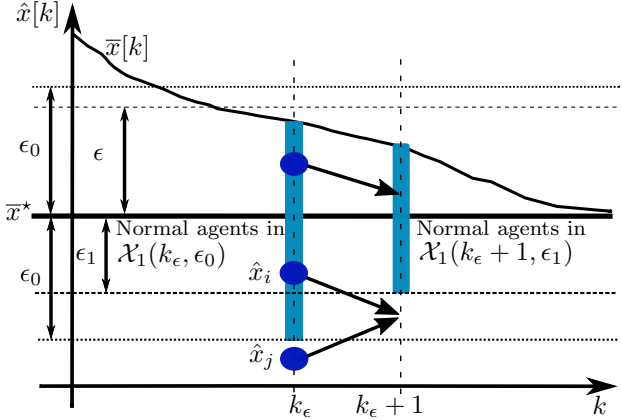
	\label{fig: SufSyncconproof}
	\vspace*{-35mm}
	\caption{A sketch for the proof of Theorem \ref{NecSufSynch}, where a set of normal agents around $\overline{x}^{\star}$ becomes smaller as time proceeds.}
\end{figure} 

We start by considering $\mathcal{X}_1(k_\epsilon,\epsilon_0 )$
(see Fig.~\ref{fig: SufSyncconproof}).
Because of the limit $\overline{x}^{\star}$,
one or more normal agents are contained
in this set or the set $\mathcal{X}_1(k_\epsilon +1,\epsilon_1 )$ from the next step. 
In what follows, we prove that $\mathcal{X}_1(k_\epsilon,\epsilon_0 )$ 
is nonempty. We can in fact show 
that each normal agent $j$ outside of $\mathcal{X}_1(k_\epsilon,\epsilon_0 )$ at time 
$k_{\epsilon}$ will remain outside of 
$\mathcal{X}_1(k_\epsilon +1,\epsilon_1 )$. 
Here, agent $j$ satisfies 
$\hat{x}_j^N[k_\epsilon] \leq \overline{x}^\star -\epsilon _0$.
From Remark~\ref{remark1}, each normal agent updates 
its position by taking a convex combination of the neighbors' 
positions at the current and previous time steps. 
Thus, by the choice of $\beta$, the position of agent $j$ 
is bounded as
\begin{align}
\hat{x}_j^N[k_\epsilon +1] 
& \leq (1-\beta)\overline {x}[k_{\epsilon }]+
    \beta (\overline{x}^{\star}- \epsilon  _0) \nonumber\\ 
& \leq (1-\beta ) ( \overline{x}^{\star}+\epsilon )
    +\beta (\overline{x}^{\star}-\epsilon_0 )\nonumber\\ 
& \leq  \overline{x}^{\star}- \beta \epsilon_0 + (1-\beta)\epsilon
 =\overline{x}^\star-\epsilon_1.
\label{eqn: normalagentbounds}
\end{align} 
Hence, agent $j$ is not in $\mathcal{X}_1(k_\epsilon +1,\epsilon_1 )$.  
Similarly, we can show that $\mathcal{X}_2(k_\epsilon,\epsilon_0)$
is nonempty.

Since the graph is $(f+1,f+1)$-robust, one of the conditions 1--3 from Definition \ref{(def:robustgraph)} holds:
\begin{enumerate}
\item[1.]
All agents in $\mathcal{X}_1(k_\epsilon,\epsilon_0)$ have 
$f+1$ neighbors from outside of 
$\mathcal{X}_1(k_\epsilon,\epsilon_0)$.
\item[2.]
All agents in $\mathcal{X}_2(k_\epsilon,\epsilon_0)$ have $f+1$ neighbors 
from outside of $\mathcal{X}_2(k_\epsilon,\epsilon_0)$.
\item[3.]
In the two sets in total, there are $f+1$ agents having at least $f+1$ neighbors from outside their own sets.
\end{enumerate}
In the first case, based on the definition of $\overline{x}^{\star}$ and $\overline{x}[k]$ there exists at least one normal agent inside the set $\mathcal{X}_1(k_\epsilon,\epsilon_0)$ having $f+1$ incoming links from outside of $\mathcal{X}_1(k_\epsilon,\epsilon_0)$. 
Similarly, in the second case, there exists one normal agent having $f+1$ incoming links from outside $\mathcal{X}_2(k_\epsilon,\epsilon_0)$. 
In the third case,
because the maximum number of malicious agents is $f$, there is one normal agent in $\mathcal{X}_1(k_\epsilon,\epsilon_0)$ or $\mathcal{X}_2(k_\epsilon,\epsilon_0)$ which has $f+1$ neighbors from outside the set it belongs to. 

Suppose that the set $\mathcal{X}_1(k_\epsilon,\epsilon_0)$ has this normal agent; denote its index by $i$. Because at each time, each normal agent ignores at most $f$ smallest neighbors and all of these $f+1$ neighbors are upper bounded by $\overline{x}^{\star}-\epsilon_0$, at least one of the agents affecting $i$ has a value smaller 
than or equal to $\overline{x}^{\star}-\epsilon_0$. From Remark \ref{remark1}, every normal agent
updates by a convex combination of position values of current and previous time steps. 
By \eqref{eqn:alpha_T} and the choice of $\beta$, 
one of the normal neighbors must be used in the update with DP-MSR.
%
Thus, for agent $i$, 
\begin{align*}
\hat{x}_i^N[k_\epsilon +1] 
& \leq (1-\beta)\overline {x}[k_{\epsilon }]
+ \beta (\overline{x}^\star- \epsilon  _0). 
\end{align*}
By \eqref{eqn: normalagentbounds}, 
$\hat{x}_i^N[k_{\epsilon}+1]$ 
is upper bounded by $\overline{x}^\star- \epsilon_1$. 
%
Thus, at least, one of the normal agents in $\mathcal{X}_1(k_\epsilon,\epsilon_0)$ has decreased to $\overline{x}^{\star} -\epsilon_1$, and the number of normal agents in $\mathcal{X}_1(k_\epsilon +1,\epsilon_1)$ is less than $\mathcal{X}_1(k_\epsilon,\epsilon_0)$ (see Fig. 2). The same argument holds for 
$\mathcal{X}_2(k_\epsilon,\epsilon_0)$.
Hence, it follows that the number of normal agents in $\mathcal{X}_1(k_\epsilon+\ell,\epsilon_\ell)$ 
is less than that in $\mathcal{X}_1(k_\epsilon+\ell-1,\epsilon_{\ell-1})$ and/or 
the number of normal agents in  $\mathcal{X}_2(k_\epsilon+\ell,\epsilon_\ell)$ 
is less than that in $\mathcal{X}_2(k_\epsilon+\ell-1,\epsilon_{\ell-1})$.
Because the number of normal agents is finite, there is a time $\ell^\star \leq n_N$ 
where the set of normal agents in $\mathcal{X}_1(k_\epsilon +\ell ,\epsilon_\ell)$ 
and/or that in $\mathcal{X}_2(k_\epsilon + \ell ,\epsilon_\ell)$ is empty for $\ell \geq \ell^\star$. 
This fact contradicts the existence of the two limits $\overline{x}^{\star}$ and $\underline{x}^{\star}$.  
Thus, we conclude that $\overline{x}^{\star}=\underline{x}^{\star}$, i.e., 
all normal agents reach position consensus.

It is finally shown that all normal agents stop asymptotically, 
which is agreement in the velocity values. When the normal agents reach agreement in their positions, the controls \eqref{eq3} become
$u_i^N[k]\to -\alpha v_i^N[k]$ as $k\to \infty$.
By the dynamics \eqref{eqn: SecondOrderDynamicsDicrete} of the agents, 
it holds that
$\hat{x}_i^N[k+1]
\to \hat{x}_i^N[k] + T \left(1 - \alpha T/2\right) v_i^N[k]$
as $k\to \infty$.
Noting \eqref{eqn:alpha_T},
we arrive at $v_i^N[k] \to 0$ as $k\to \infty$.
%
\end{proof}

In \cite{dibajiishiiSCL2015}, we have studied
the $f$-local model, where each normal agent has 
at most $f$ malicious agents as neighbors.
Clearly, there may be more malicious agents overall
than the $f$-total case. 
There, a sufficient condition for resilient consensus
is that the network is $(2f+1)$-robust. From
Lemma~\ref{lemma:robust graphs}, such graphs require more edges 
than $(f+1,f+1)$-robust graphs, as given
in the theorem above. 

The result is consistent with that for the first-order agent case in \cite{zhang1}. 
More from the technical side,
difficulties in dealing with second-order agent dynamics
can be described as follows:
In the proof above, an important step is to establish that
$\overline{x}[k]$ and $\underline{x}[k]$ defined in 
\eqref{eqn: SecondOrderVariables} are monotonically nonincreasing and 
nondecreasing, respectively. These properties do not hold
for the maximum position $\max\, \hat{x}^N[k]$ and the minimum 
position $\min\, \hat{x}^N[k]$ as in the first-order case.
Furthermore, as a consequence of this fact,
it is more involved 
to show that the sets
$\mathcal{X}_1(k_\epsilon +{\ell},\epsilon_{\ell} )$ and
$\mathcal{X}_2(k_\epsilon+{\ell},\epsilon_{\ell})$ 
in \eqref{eqn: Thedelayedmaxminsets} 
become smaller as ${\ell}$ increases.

As a corollary to Theorem~\ref{NecSufSynch}, we show that 
the convergence rate to achieve consensus is exponential. 

\begin{corollary}\label{cor:Synch}\rm
Under the assumptions in Theorem~\ref{NecSufSynch}, 
the network of agents with second-order dynamics 
reaches resilient consensus 
with an exponential convergence rate. 
\end{corollary}

{\color{black}
\begin{proof}
We outline the proof which follows a similar line of argument 
to that of Theorem~\ref{NecSufSynch},
but with the knowledge that the agents come to agreement. 
Let $V(k)=\overline{x}[k]-\underline{x}[k]$. We show that
this function decreases to zero exponentially fast as $k\rightarrow\infty$.

Take an arbitrary constant $\eta\in(0,1)$.
We introduce two sets as follows: For $k=0,1,\ldots,n_N$, let
\begin{align*}
 \mathcal{X}'_1(k) 
  &= \big\{
        j\in\mathcal{V}:~
         \hat{x}_j[k] 
            > \big(1-\beta^k \eta \big)
               \big(\overline{x}[0] - x^{\star}\big) + x^{\star}
     \big\},\\
 \mathcal{X}'_2(k) 
  &= \big\{
        j\in\mathcal{V}:~
         \hat{x}_j[k]
            < \beta^k \eta 
               \big(\underline{x}[0] - x^{\star}\big) + \underline{x}[0]
     \big\}.
\end{align*}
Clearly, for each $k$, the sets $\mathcal{X}'_1(k)$ and $\mathcal{X}'_2(k)$
are disjoint, so by $(f+1,f+1)$ robustness, there is one normal
agent $i$ in one of the sets having $f+1$ incoming links from outside
the set to which it belongs. If agent $i$ is in $\mathcal{X}'_1(k)$, 
then, similar to \eqref{eqn: normalagentbounds}, 
we can upper bound its position as
\begin{equation*}
  \hat{x}_i[k+1]
    \leq (1-\beta^{k+1}) \overline{x}[0]
           + \beta^{k+1} \big[(1-\eta)(\overline{x}[0]-x^\star)+x^\star\big].
\end{equation*}
Hence, in this case, agent $i$ is not in $\mathcal{X}'_1(k+1)$ at time $k+1$. 
On the other hand, if agent $i$ is in $\mathcal{X}'_2(k)$, then
\begin{equation*}
  \hat{x}_i[k+1]
    \geq (1-\beta^{k+1}) \underline{x}[0]
           + \beta^{k+1} \big[\eta(x^\star - \overline{x}[0]) 
                              + \overline{x}[0]\big],
\end{equation*}
implying that $\hat{x}_i[k+1]$ is outside of $\mathcal{X}'_2(k+1)$.
Since the number of normal agents is $n_N$, at time $k=n_N$,
both of the sets $\mathcal{X}'_1(n_N)$ and $\mathcal{X}'_2(n_N)$ do
not contain any normal agent. Hence, we can conclude that the 
maximum position $\overline{x}[k]$ and the minimum position 
$\underline{x}[k]$ 
of normal agents, respectively, decreased and increased since time 0.
More concretely, we have $V(n_N)\leq (1-\beta^{n_N}\eta)V(0)$.
By repeating this argument, we can establish that 
$V(k n_N)\leq (1-\beta^{n_N}\eta)^kV(0)$.
\end{proof}

This corollary also indicates that conventional consensus algorithms 
without malicious agents in the network have exponential convergence
rate. In the proposed MSR-type algorithm, the rate of convergence is affected
by two factors. One is that a number of edges are ignored and
not used for the updates, which will reduce the convergence rate. 
The other is that if malicious agents stay together with some 
of the normal agents, they can still influence the rate to achieve
consensus and slow it down.
}

\section{Networks with {\color{black}Partial} Asynchrony and Delay}\label{sect: asynchsystems}

So far, the underlying assumption in the model has been 
that all agents exchange their states at the same time instants. 
Moreover, they make updates based on the relative
locations of their neighbors without any time delay in the model.
In practice, however, the agents might not be synchronized
nor have access to the current data of all neighbors simultaneously. 
In this section, we extend the setup so that the agents are 
allowed to update at different times with
delayed information. 


We would like to emphasize the difference in the {\color{black}partially} asynchronous 
agent model employed here from those in the 
resilient consensus literature.
In particular, we follow the approach generally assumed in 
asynchronous consensus for the case without malicious agents; 
see, e.g., \cite{mesbahi,Yangfeng,Xiao2006} for single-integrator networks 
and \cite{LinP,Cheng-Lin,JiahuQin2012,JiahuHirche} for the double-integrator case.
That is, at the time for an update, each agent uses the most recently received positions
of its neighbors. This is a natural setting especially for autonomous
mobile robots or vehicles using sensors to locate their neighbors
in real time.

In contrast, the works \cite{Azadmanesh2002,Azadmanesh1993,LeBlancPaper,Vaidya} from the area of computer science consider asynchronous MSR-type algorithms based on the notion of \textit{rounds}
(for the case with first-order agents).
There, when each agent makes a transmission, it broadcasts its state together with its round $r$, 
representing the number of transmissions made so far.
The agent makes an update to obtain the new state corresponding to round $r+1$ by 
using the states of neighbors, but only when a sufficient number of those labeled 
with the same round $r$ are received. 
Due to delays in communication, 
the states labeled with round $r$ may be received at
various times, causing potentially large delays in making the $(r+1)$th update for
some agents. 

Compared to the results in the previous section,
the analysis in the {\color{black}partially} asynchronous model studied here 
becomes more complicated.
Moreover, the derived condition is more restrictive because
there are additional ways for the malicious 
agents to deceive the normal ones.
For example, they may quickly move so that they appear to be at different positions 
for different normal agents, which may prevent them from coming together. 

\subsection{Asynchronous DP-MSR Algorithm}\label{Subsec: AsynchFormulation}

Here, we employ the control input taking account of 
possible delays in the position values from 
the neighbors as 
\begin{equation}
{\color{black}
  u_i[k]
   = \sum_{j \in \mathcal{N}_i} 
       a_{ij}[k]\bigl(
                  \hat{x}_j[k-\tau_{ij}[k]]-\hat{x}_i[k]
                \bigr) -\alpha  v_i[k], 
\label{eqn: DelayControl1}}
\end{equation}
where $\tau_{ij}[k] \in \mathbb{Z}_+$ denotes the delay 
in the edge $(j,i)$ at time $k$. 
From the viewpoint of agent $i$, the most 
recent information regarding agent $j$ 
at time $k$ is the position at time $k-\tau_{ij}[k]$ relative 
to its own current position.
The delays are time varying and may be different at each edge, 
but we assume the common upper bound $\tau$ as
\begin{equation}\label{eqn:delay_tau}
 0 \leq \tau_{ij}[k] \leq \tau,~~
  (j,i)\in\mathcal{E},~k\in\mathbb{Z}_+.
\end{equation}
Hence, each normal agent becomes aware of 
the position of each of its neighbors at least 
once in $\tau$ time steps, but possibly at different
time instants. 
{\color{black}In other words, normal agents must update and 
transmit their information often enough to meet \eqref{eqn:delay_tau}.}
It is also assumed in \eqref{eqn: DelayControl1} that 
agent~$i$ uses its own current velocity. 
%
We emphasize that the value of $\tau$ in \eqref{eqn:delay_tau}
can be arbitrary and moreover
need not be known to the agents since this information is not used
in the update rule.
In \cite{Gao,Cheng-Lin,JiahuHirche}, 
time delays for {\color{black} partially} asynchronous cases 
have been studied for agents with second-order dynamics. 

As in the synchronous case, the malicious agents are assumed 
to be omniscient. Here, it means that they have prior knowledge 
of update times and $\tau_{ij}[k]$ for all links and $k \geq 0$. 
The malicious agents might take advantage of this 
knowledge to decide how they should make updates to 
confuse and mislead the normal agents.

To achieve resilient consensus, we employ a modified version
of the algorithm in Section~\ref{sect:SynchSystems},
called the asynchronous DP-MSR, outlined below. 
\begin{enumerate}
\item[1.] At each time step $k$, each normal agent $i$ 
decides whether to make an update or not. 
\item[2.]
If it decides to do so, then it uses the relative position values
of its neighbors $j\in\mathcal{N}_i$
based on the most recent values in the form of 
$\hat{x}_j[k-\tau_{ij}[k]]-\hat{x}_i[k]$ 
and then follows step~2
of the DP-MSR algorithm based on these values.
Afterwards, it applies the control input 
\eqref{eqn: DelayControl1} by substituting $a_{ij}[k]=0$
for edges $(j,i)$ which are ignored in step~2 of DP-MSR.
\item[3.] Otherwise, it applies the control \eqref{eqn: DelayControl1} 
where the first term of position values of its neighbors 
remains the same as the previous time step,
and for the second term, its own current 
velocity is used. 
\end{enumerate}

The asynchronous version of the resilient consensus problem is stated
as follows: Under the $f$-total malicious model, find a condition on the
network topology so that the normal agents reach resilient consensus using the asynchronous DP-MSR algorithm. 

\subsection{Matrix Representation}\label{subsec: MatrixRepresenation}

Before presenting the main result of this section, we introduce some notation to represent the equations in the matrix form. Define the matrices $A_{\ell}[k]$, 
$ 0 \leq \ell \leq \tau$, by
\begin{align*}
 &\big(A_{\ell}[k]\big)_{ij}
    =\begin{cases} 
         a_{ij}[k] &
           \text{if $(j,i)\in\mathcal{E}[k]$ and
              $\tau_{ij}[k]=\ell$},\\
         0 & \text{otherwise}.
\end{cases}
\end{align*}
Then, let $D[k]$ be a diagonal matrix whose $i$th entry is given by
$d_i[k] = \sum_{j=1}^{n} a_{ij}[k]$.
Now, the $n \times (\tau +1 ) n$ matrix $L_\tau[k]$ is defined as
\begin{equation*} 
  L_\tau [k]
   = \begin{bmatrix}
      D[k]-A_0[k]~~\cdots~~~-A_\tau [k]
     \end{bmatrix}.
\end{equation*}
It is clear that the summation of each row is zero as in the Laplacian matrix $L[k]$. 

Now, the control input \eqref{eqn: DelayControl1} can be 
expressed as 
\begin{equation} \label{eqn: DelayControl3}
\begin{split}
  u^N[k]
   &= -L^N_\tau[k]z[k]
        -\alpha
           \begin{bmatrix}
             I_{n_N} & 0
           \end{bmatrix} v[k],\\
  u^M[k]&:~\text{arbitrary},
\end{split}
\end{equation} 
where $z[k]=[\hat{x}[k] ^T \hat{x}[k-1]^T \cdots \hat{x}[k-\tau]^T]^T$ 
is a $(\tau +1)n$-dimensional vector for $k \geq 0$ and $L^N_{\tau}[k]$ is 
a matrix formed by the first $n_N$ rows of $L_\tau[k]$. 
Here, $z[0]$ is the given initial position values of the agents {\color{black}and can be chosen arbitrarily.}
By \eqref{eqn:rv} and \eqref{eqn: DelayControl3}, 
the agent dynamics can be written as
\begin{align} \label{eqn: DelayPosVel}
\hat{x}[k+1]
  &= \Gamma[k] z[k] + Qv[k] + \nonumber\frac{T^2}{2}\begin{bmatrix}
0 \\ I_{n_M}
\end{bmatrix} u^M[k], \\
v[k+1]&=-T\begin{bmatrix}
L^N_{\tau}[k] \\ 0
\end{bmatrix}z[k]+Rv[k]+T\begin{bmatrix}
0 \\ I_{n_M}
\end{bmatrix}u^M[k],
\end{align}
where $\Gamma[k]$ is an $n \times (\tau+1)n$ matrix given by
\begin{equation*}
\Gamma[k]=[I_n~~0]-\frac{T^2}{2}\begin{bmatrix} L^N_\tau[k]  \\ 0
\end{bmatrix}
\end{equation*}
and $Q$ and $R$ are given in \eqref{eqn: MatrixNotation}. 
Based on the expression \eqref{eqn: DelayPosVel}, we can derive 
a result corresponding to Lemma~\ref{lemmapositionupdates} 
for the {\color{black}partially} asynchronous and delayed protocol. 
The proof is omitted since it is by direct calculation similar to 
Lemma~\ref{lemmapositionupdates} shown in \cite{dibajiishiiSCL2015}.

\begin{lemma}\label{Lemma: DelayFinalEqu}\rm
Under the control inputs \eqref{eqn: DelayControl3}, 
the position vector $\hat{x}[k]$ of the agents for $k \geq 1$ can be expressed as
\begin{align*}
  \hat{x}[k+1]
     &=\bigl[
          \Lambda_{1k}~~ \Lambda_{2k}
       \bigr]
         \begin{bmatrix}
	z[k]\\
	z[k-1]
	\end{bmatrix}\nonumber\\
     &\hspace*{8mm}\mbox{}+\frac{T^2}{2}\begin{bmatrix}
	0 \\ I_{n_M}
	\end{bmatrix}\bigl(u^M[k]+u^M[k-1]\bigr),
\end{align*}
	where 
\begin{align*}
	\Lambda_{1k}&= \bigl[R ~~0 \bigr]+\Gamma[k],~~
	\Lambda_{2k}= -R\,\Gamma[k-1]-QT\begin{bmatrix}   
	L^N_\tau[k] \\ 0
	\end{bmatrix}.
\end{align*}
	Furthermore, the matrix 
$\bigl[\Lambda_{1k}~~\Lambda_{2k}\bigr]$ is nonnegative and the sum of each of its first $n_N$ rows is one. 
\end{lemma}

\subsection{Resilient Consensus Analysis}\label{subsec: DelayRCAnalysis}

The following theorem is the main result of the paper,
addressing resilient consensus via the asynchronous DP-MSR 
in the presence of delayed information. 
The safety interval differs from the previous 
case and is given by
\begin{align}\label{eqn:  DelaySafetyInterval}
 \mathcal{S_\tau}
   &= \biggl[
         \min z^N[0]
           + \min \biggl\{0,
                    \biggl(
                      T - \frac{\alpha T^2}{2}
                    \biggr) v^N[0]
                  \biggr\},\notag\\
   &\hspace*{0.5cm}
      \max z^N[0]
        + \max \biggl\{0,
                 \biggl(
                   T - \frac{\alpha T^2}{2}
                 \biggr) v^N[0]
               \biggr\}
       \biggr].
\end{align}

\begin{theorem}\label{Theorem: DelaySuf}\rm
Under the $f$-total malicious model, 
the network of agents with second-order dynamics using 
the control in \eqref{eqn: DelayControl3} and 
the asynchronous DP-MSR algorithm reaches resilient consensus
only if the underlying graph $(f+1,f+1)$-robust.
Moreover, if the underlying graph is $(2f+1)$-robust,
then resilient consensus is attained 
with a safety interval given by \eqref{eqn:  DelaySafetyInterval}.	
\end{theorem}

{\color{black}
The proof of this theorem given below follows an argument 
similar to that of Theorem~\ref{NecSufSynch}. 
However, the problem is more general with the delay bound
$\tau \geq 0$ in \eqref{eqn:delay_tau}. 
This in turn results in more involved 
analysis with subtle differences. 
We provide further discussions later. }

\begin{proof} 
(Necessity)~The synchronous network is a special case of {\color{black}partially} 
asynchronous ones with $\tau=0$. Thus, the necessary condition
in Theorem~\ref{NecSufSynch} is valid. 

(Sufficiency)~We first show that the safety condition holds. 
For $k=0$, by \eqref{eqn:  DelaySafetyInterval}, 
we have $\hat{x}_i[0]\in\mathcal{S}_{\tau}$ for $i\in\mathcal{V}\setminus\mathcal{M}$.
For $k=1$, by \eqref{eqn: DelayPosVel},
the positions of normal agents can be expressed as
\begin{equation}\label{eqn: Delay:k=1 update}
	\hat{x}^N[1]
	= \biggl(
	\begin{bmatrix}
	I_{n_N} & 0
	\end{bmatrix} - \frac{T^2}{2}L^N_\tau [0]
	\biggr) z[0] 
	+ \biggl(
	T-\frac{\alpha T^2}{2}
	\biggr) v^N[0].	
\end{equation}
This vector may be affected by the malicious agents 
through their initial positions. 
However, 
by step~2 in DP-MSR, 
for any normal agent, if some 
neighbors are malicious and are outside of $[\min z^N[0],\max z^N[0]]$,
then they will be ignored. 
In \eqref{eqn: Delay:k=1 update}, the matrix $[I_{n_N} 0]- (T^2/2)L_\tau ^N[0]$ is 
nonnegative and its row sums are one. 
Hence, the first term on the right-hand side of 
\eqref{eqn: Delay:k=1 update} becomes
convex combinations of values in the interval
$[\min z^N[0],\max z^N[0]]$.
Thus, we have
$\hat{x}_i[1]\in\mathcal{S_\tau}$ for $i\in\mathcal{V}\setminus\mathcal{M}$.
	
	
Next, for $k \geq 1$, define two variables by
\begin{equation}\label{eqn: DelayVarDefine}
\begin{split}
\overline{x}_{\tau}[k]
 &=\max \left(
          \hat{x}^N[k],\hat{x}^N[k-1],\ldots,\hat{x}^N[k-\tau-1]
         \right),\\
\underline{x}_{\tau} [k]
 &=\min\left(
         \hat{x}^N[k],\hat{x}^N[k-1],\ldots,\hat{x}^N[k-\tau-1]
       \right).
\end{split}
\end{equation} 
Here, we claim that $\overline{x}_{\tau}[k]$ is a nonincreasing 
function of $k\geq 1$. By Lemma \ref{Lemma: DelayFinalEqu}, 
at time $k \geq 2$, each normal agent updates its position 
based on a convex combination of the neighbors' positions  
from $k-1$ to $k-\tau-1$. 
If some neighbors are malicious and stay outside of
the interval determined by the normal agents' positions,
then they are ignored in step~2 of DP-MSR. 
Hence, we obtain
$\hat{x}_i[k+1] 
\leq \max \left (
            \hat{x}^N[k],\hat{x}^N[k-1],\ldots,\hat{x}^N[k-\tau-1]
          \right)$
for $i\in\mathcal {V}\setminus\mathcal {M}$.
It also follows that
\begin{align*}
  \hat{x}_i[k]
    &\leq \max \left (\hat{x}^N[k],\ldots,\hat{x}^N[k-\tau-1]\right),\\
  \hat{x}_i[k-1]
    &\leq \max \left (\hat{x}^N[k],\ldots,\hat{x}^N[k-\tau-1]\right),\\
	\qquad\qquad\quad&~~\vdots\\
  \hat{x}_i[k-\tau]
    &=\hat{x}_i[k+1-(\tau +1)] \\
	~~\quad\quad\quad &\leq \max \left (\hat{x}^N[k],\ldots,\hat{x}^N[k-\tau-1]\right)
	\end{align*}
	for $i\in\mathcal {V}\setminus\mathcal {M}$. Hence, we have 
	\begin{align*}
	\overline{x}_{\tau}[k+1]&=\max\left(\hat{x}^N[k+1],\ldots,\hat{x}^N[k+1-(\tau+1)]\right)\\
	&\leq \max \left(\hat{x}^N[k],\ldots,\hat{x}^N[k-(\tau+1)]\right)=\overline{x}_{\tau}[k].
	\end{align*}
We can similarly prove 
that $\underline{x}_{\tau}[k]$ is nondecreasing in time. 
This indicates that for $k \geq 2$, we have $\hat{x}_i[k] \in \mathcal{S_\tau}$ for $i \in \mathcal{V}\backslash \mathcal{M}$. Thus, we have shown the safety condition.

	In the rest of the proof, we must show the agreement condition.
	As $\overline{x}_{\tau}[k]$ and $\underline{x}_{\tau}[k]$ are monotone functions 
	and contained in $ [\underline{x}_{\tau}[1],\overline{x}_{\tau}[1]]$, 
        both of their limits exist, which are
        denoted by $\overline{x}^\star_{\tau}$ and 
        $\underline{x}^\star_{\tau}$, respectively.
	We claim that the limits in fact satisfy
	$\overline{x}^\star_{\tau}=\underline{x}^\star_{\tau}$, i.e., the positions of the normal
	agents come to consensus.
	The proof is by contradiction, so we assume that $\overline{x}^\star_{\tau}>\underline{x}^\star_{\tau}$.

	First, let $\beta$ be the minimum 
	nonzero element over all possible cases of $\bigl[\Lambda_{1k}~\Lambda_{2k}\bigr]$.
	From \eqref{eqn:alpha_T} and the bound $\gamma$ on $a_{ij}[k]$, 
	it holds that $\beta\in(0,1)$.
	Choose $\epsilon _0>0$ and $\epsilon>0$ small enough that
	\begin{align}
	& \underline{x}^\star_{\tau} +\epsilon _0
	< \overline{x}^\star_{\tau}-\epsilon _0,~~
	\epsilon
	< \frac{\beta^{(\tau+1)n_N}\epsilon_0}{1-\beta ^{(\tau+1)n_N}}.
	\label{eqn: xstarDelay}
	\end{align}
	We next take the sequence $\{\epsilon_\ell\}$ via 
	\[
	\epsilon_{\ell+1}
	= \beta \epsilon_{\ell} - (1-\beta)\epsilon,~~\ell=0,1,\ldots,(\tau+1)n_{N}-1.
	\]
	It can be shown that
	$0<\epsilon_{\ell+1} < \epsilon_{\ell}$ for all $\ell$.
	In particular, they are positive because 
        by \eqref{eqn: xstarDelay}, it holds that
	\begin{align*}
	\epsilon_{(\tau+1)n_N}  
	&= \beta^{(\tau+1)n_N} \epsilon_{0}  
	- \sum_{k=0}^{(\tau+1)n_N-1}
	\beta^k (1-\beta) \epsilon\\
	&= \beta^{(\tau+1)n_N} \epsilon_{0}  
	-\bigl(
	1- \beta^{(\tau+1)n_N} 
	\bigr)\epsilon
	> 0.
	\end{align*}
	We also take $k_\epsilon\in\mathbb{Z}_{+}$ such that 
	$\overline{x}_{\tau}[k] 
	< \overline{x}^\star_{\tau} +\epsilon$ and 
	$\underline{x}_{\tau}[k]
	> \underline{x}^\star_{\tau} -\epsilon$ for $k\geq k_\epsilon$.
	Due to convergence of $\overline{x}_{\tau}[k]$ and $\underline{x}_{\tau}[k]$,
	such $k_\epsilon$ exists. 
	For the sequence $\{\epsilon_\ell\}$,
	let 
	\begin{align*}
	\mathcal{X}_{1\tau}(k_{\epsilon}+\ell,\epsilon_\ell)
	&=\{j \in \mathcal{V}\backslash \mathcal{M}:~
	\hat{x}_j[k_{\epsilon}+\ell]>\overline{x}^\star_{\tau}-\epsilon_\ell \},\\ 
	\mathcal{X}_{2\tau}(k_{\epsilon}+\ell,\epsilon_\ell)
	&=\{j \in \mathcal{V}\backslash \mathcal{M}:~ 
	\hat{x}_j[k_{\epsilon}+\ell]<\underline{x}^\star_{\tau}+\epsilon_\ell \}. 
	\end{align*}
	These two sets are disjoint by \eqref{eqn: xstarDelay} and 
	$0<\epsilon_{\ell+1} < \epsilon_{\ell}$.

	Next, we must show that 
	one of the two sets becomes empty in a finite number 
        of steps. This clearly contradicts
        the assumption on $\overline{x}^\star_{\tau}$ 
	and $\underline{x}^\star_{\tau}$ being the limits.
	Consider the set $\mathcal{X}_{1\tau}(k_\epsilon,\epsilon_0 )$.
	Due to the definition of $\overline{x}_{\tau}[k]$ in \eqref{eqn: DelayVarDefine} and its limit $\overline{x}_{\tau}^\star$, 
one or more normal agents are contained in the union of the sets 
$\mathcal{X}_{1\tau}(k_\epsilon +\ell,\epsilon_\ell)$ 
for $0\leq \ell \leq \tau+1$. 
We claim that $\mathcal{X}_{1\tau}(k_\epsilon,\epsilon_0 )$ is in fact 
nonempty. To prove this, 
it is sufficient to show that if a normal agent $j$ is not in 
$\mathcal{X}_{1\tau}(k_\epsilon+\ell, \epsilon_{\ell})$, 
then it is not in $\mathcal{X}_{1\tau}(k_\epsilon+\ell+1, \epsilon_{\ell+1})$
for $\ell =0,\ldots,\tau$.

Suppose agent $j$ satisfies
$\hat{x}_j[k_\epsilon+\ell] \leq \overline{x}_\tau^\star - \epsilon_{\ell}$.
From Lemma~\ref{Lemma: DelayFinalEqu}, every normal agent updates
its position to a convex combination of the neighbors' position values 
at the current or previous times. 
Though the neighbors may be malicious here,
the ones at positions greater than $\overline{x}_{\tau}[k_{\epsilon }+\ell]$
are ignored in step~2 of DP-MSR.
Hence, the position of agent $j$ at the next time step 
is bounded as
\begin{align}
  &\hat{x}_j[k_\epsilon +\ell+1] 
    \leq (1-\beta)\overline{x}_{\tau}[k_{\epsilon }+\ell]
  	  + \beta (\overline{x}^\star_{\tau}- \epsilon _\ell) \nonumber\\ 
 &\hspace*{1cm} 
    \leq (1-\beta ) (\overline{x}^\star_{\tau}+\epsilon )
	+ \beta (\overline{x}^\star_{\tau}-\epsilon_\ell )\nonumber\\ 
 &\hspace*{1cm} 
    \leq  \overline{x}^\star_{\tau}- \beta \epsilon_\ell 
                  +(1-\beta)\epsilon
    =  \overline{x}^\star_{\tau} - \epsilon_{\ell+1}.
\label{eqn: Delaynormalagentbounds}
\end{align}
It thus follows that agent $j$ is not in $\mathcal{X}_{1\tau}(k_\epsilon+\ell +1,\epsilon_{\ell +1} )$.  
	This means that the cardinality of the set $\mathcal{X}_{1\tau}(k_\epsilon+\ell,\epsilon_\ell)$ is nonincreasing for $\ell=0,\ldots,\tau+1$. 
	The same holds for 
	$\mathcal{X}_{2\tau}(k_\epsilon+\ell,\epsilon_\ell)$. 
	
	We next show that one of these two sets in fact becomes empty in finite time. 
	By $(2f+1)$-robustness, between the two nonempty disjoint sets 
	$\mathcal{X}_{1\tau}(k_\epsilon,\epsilon _{0})$ and 
	$\mathcal{X}_{2\tau}(k_\epsilon,\epsilon _0)$, 
	one of them has an agent with $2f+1$ incoming links from outside the set. 
	Suppose that $\mathcal{X}_{1\tau}(k_{\epsilon},\epsilon _0)$ 
	has this property and let $i$ be the agent in this set which has $2f+1$ neighbors 
	outside $\mathcal{X}_{1\tau}(k_ \epsilon ,\epsilon _0)$. 
	Since there are at most $f$ malicious neighbors for this normal agent, 
	there are at least $f+1$ normal neighbors outside 
	$\mathcal{X}_{1\tau}(k_ \epsilon ,\epsilon _0)$;
	by the argument above, these normal agents
	will not be in $\mathcal{X}_{1\tau}(k_\epsilon+\ell ,\epsilon _{\ell})$ 
	for $0\leq\ell\leq \tau$.
	By step~2 of DP-MSR, one of these normal neighbors
	must be used in the updates of agent $i$ at any time. 
	In particular, when agent $i$ makes an update at time $k_{\epsilon}+\tau$, 
	a normal agent's delayed position is used, upper bounded by
	$\overline{x}^\star_{\tau}-\epsilon _{\tau}$.
	%
	It thus follows that, at time $k_{\epsilon}+\tau$, when agent $i$ makes an update, 
	its position can be bounded as
	\begin{align*}
	\hat{x}_i[k_\epsilon +\tau+1] 
	& \leq (1-\beta)\overline {x}_{\tau}[k_{\epsilon }+\tau]
	+ \beta (\overline{x}^\star_{\tau}- \epsilon  _{\tau}).
	\end{align*}
	By \eqref{eqn: Delaynormalagentbounds}, we have
	$\hat{x}_i[k_{\epsilon}+\tau+1]\leq\overline{x}^\star_{\tau}- \epsilon_{\tau+1}$. 
	We can conclude that if agent $i$ in $\mathcal{X}_{1\tau}(k_ \epsilon,\epsilon_{0})$ 
	has $2f+1$ incoming links from outside the set, then 
	it goes outside of 
	$\mathcal{X}_{1\tau}(k_\epsilon + \tau + 1,\epsilon_{\tau+1})$ 
        after $\tau+1$ steps.
	Consequently, $\mathcal{X}_{1\tau}(k_\epsilon +\tau+1,\epsilon_{\tau+1})$ 
        has the cardinality 
	smaller than that of $\mathcal{X}_{1\tau}(k_\epsilon,\epsilon _{0})$, that is,
	$\abs{\mathcal{X}_{1\tau}(k_\epsilon+\tau+1,\epsilon_{\tau+1})} 
	< \abs{\mathcal{X}_{1\tau}(k_\epsilon,\epsilon_{0})}$.
	Likewise, it follows that if $\mathcal{X}_{2\tau}(k_\epsilon ,\epsilon_{0})$ 
	has an agent with $2f+1$ incoming links from the rest, then
	$\abs{\mathcal{X}_{2\tau}(k_\epsilon +\tau+1,\epsilon_{\tau+1})} 
	< \abs{\mathcal{X}_{2\tau}(k_\epsilon,\epsilon_{0})}$.
	
Since there are only a finite number $n_N$ of normal agents,
we can repeat the steps above until one of the sets
$\mathcal{X}_{1\tau}(k_\epsilon +\tau+1,\epsilon _{\tau+1})$ and 
$\mathcal{X}_{2\tau}(k_\epsilon +\tau+1,\epsilon _{\tau+1})$ 
becomes empty; it takes no more than $(\tau+1)n_N$ steps.
Once the set becomes empty, it will remain so indefinitely.
%
%
This contradicts the assumption that $\overline{x}^\star_{\tau}$ 
and $\underline{x}^\star_{\tau}$ are the limits.
Therefore, we obtain $\overline{x}^\star_{\tau}=\underline{x}^\star_{\tau}$.
\end{proof}

{\color{black}
It is interesting to note that the bound $\tau$ on the delay 
time in \eqref{eqn:delay_tau}
can be arbitrary and moreover need not be known
to any of the agents. Hence, in this respect, 
the condition~\eqref{eqn:delay_tau} is not restrictive. 
A trade-off concerning delay is that longer delays will 
result in slower convergence.
This property can be explained in the proof above as follows.
We observe 
from the monotonicity of $\overline{x}_{\tau}[k]$ and
$\underline{x}_{\tau}[k]$ in \eqref{eqn: DelayVarDefine}
that the normal agents come together if we look at the 
$\tau+1$ time step horizon in the past. 
Finally, note that the convergence rate of the asynchronous update case
can be shown to be exponential as 
in the synchronous case.
}

Theorem~\ref{Theorem: DelaySuf} demonstrates that 
to achieve resilient consensus
in the {\color{black}partially} asynchronous delayed setting requires 
a more dense graph than that in the synchronous setting
in Theorem~\ref{NecSufSynch}.
Indeed, from Lemma~\ref{lemma:robust graphs},
a graph that is $(2f+1)$-robust is also $(f+1,f+1)$-robust.
On the other hand, the difference in graph robustness 
appearing in the two theorems have certain effects on the proofs. 
For Theorem~\ref{Theorem: DelaySuf}, 
the sets $\mathcal{X}_{1\tau}[k]$ and 
$\mathcal{X}_{2\tau}[k]$ do not include 
the malicious agents, while the sets $\mathcal{X}_{1}[k]$ 
and $\mathcal{X}_{2}[k]$ in the proof of 
Theorem~\ref{NecSufSynch} involve both normal and 
malicious agents. This difference originates from 
the definitions of $(2f+1)$- and $(f+1,f+1)$-robust graphs.
In fact, in the $f$-total model,
we see that the second $f+1$ in 
$(f+1,f+1)$ guarantees that at least one of 
the agents in $\mathcal{X}_{1}[k]$ or $\mathcal{X}_{2}[k]$ 
is normal and has enough number of incoming links for convergence. 
However, $(2f+1)$-robustness is a more 
local notion. Since the worst-case behavior of the 
malicious agents happens in the neighborhood of each 
normal agent, the sets $\mathcal{X}_{1\tau}[k]$ and 
$\mathcal{X}_{2\tau}[k]$ are defined in this way.





In the above result, we observe that there is a gap between the sufficient condition 
and the necessary condition. 
However, this gap may be essential to the problem. 
{\color{black}
To illustrate this point, we present 
a $2f$-robust graph in Fig.~\ref{fig: AsyncProp}, 
which is not resilient to $f$ totally bounded faults as we show formally. 

This graph is composed of four subgraphs 
$\mathcal{G}_i$, $i=1,\ldots,4$, 
and each of them is a complete graph.
The graph $\mathcal{G}_1$ consists of $4f$ agents and the rest
have $f$ agents.
Each agent in $\mathcal{G}_2$ has incoming links from $2f$ agents of $\mathcal{G}_1$. 
Every agent in $\mathcal{G}_3$ 
has $f$ links from $\mathcal{G}_1$ and $f$ links from $\mathcal{G}_2$. 
Likewise,
each agent of $\mathcal{G}_4$ has a link from every agent in $\mathcal{G}_1$ and $f$ incoming links from $\mathcal{G}_2$. 
Note that the minimum degree for a $2f$-robust graph is $2f$. However for this graph, the minimum degree of 
the agents is $2f+1$ or greater. This is an important point for the following reason. 
If a normal agent has only $2f$ neighbors,
it might ignore all of them under the asynchronous
DP-MSR algorithm, which in turn
means that the agent will remain at its current position.
It is clear that if this happens for more than two agents in the network, consensus can not take place.
The next proposition is based on this graph.

\begin{figure}
	\centering
	\vspace*{-19mm}
	\def \svgwidth{6.7cm}
	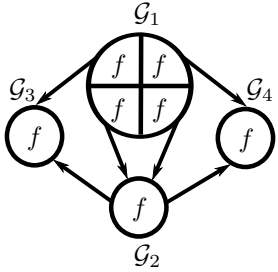
	\vspace*{-1mm}
	\caption{{\color{black}A $2f$-robust graph which fails to reach consensus with {\color{black}partial} asynchrony and delayed information.}}
\label{fig: AsyncProp}
\end{figure}
}

\begin{proposition}\label{Prop: Asynch2frobust}\rm
There exists a 2$f$-robust network with the minimum 
degree $2f+1$ under which normal agents may not achieve 
resilient consensus by the asynchronous DP-MSR algorithm. 
\end{proposition}


{\color{black}
\begin{proof}
We claim that the graph in Fig.~\ref{fig: AsyncProp} is 
$2f$-robust and $(f+1,f+1)$-robust at the same time, 
but resilient consensus cannot be reached under the asynchronous DP-MSR. 
Suppose that all agents in $\mathcal{G}_2$ are malicious. 
We show a scenario in which 
by the DP-MSR algorithm, the values of the agents in $\mathcal{G}_3$ and $\mathcal{G}_4$ never agree.
Note that $\mathcal{G}_1$ is 2$f$-robust because of Lemma \ref{lemma:robust graphs} (v). By (iv) of this lemma, the graph obtained by adding $\mathcal{G}_2$ is still 2$f$-robust, since there are 2$f$ edges from $\mathcal{G}_1$. 
Similarly, adding $\mathcal{G}_3$ and $\mathcal{G}_4$ and the 
required edges 
also keeps the graph to be 2$f$-robust. 
	
	We assume that all agents start from stationary positions. 
	Assign $(a,0)$ as the initial position and velocity 
	(for $k=0$ and the prior $\tau$ steps) 
	of the agents in $\mathcal{G}_3$ 
	and $(b,0)$ as those of $\mathcal{G}_4$. 
	All agents in $\mathcal{G}_1$ are given $(c,0)$ as their initial positions and velocities, where $a < c < b$. 
The malicious agents in $\mathcal{G}_2$ take $a$ at even time steps and $b$ at odd time steps. The time delays are chosen by the following scenario: $\tau_{ij}[2m]=0$ 
and $\tau_{ij}[2m+1]=1$ for $(j,i) \in \mathcal{E}, j \in \mathcal{V}_2$, 
$i \in \mathcal{V}_3$, and $m \in \mathbb{Z}_+$. 
Also, $\tau_{ij}[2m]=1$ and $\tau_{ij}[2m+1]=0$ for $(j,i) \in \mathcal{E}$, $j \in \mathcal{V}_2$, $i \in \mathcal{V}_4$, and $m \in \mathbb{Z}_+$. All other links have no delay. 
Then, to the agents in $\mathcal{G}_3$,
the malicious agents appear to be stationary at $a$ and 
to the agents in $\mathcal{G}_4$ at $b$.
	
By executing the asynchronous DP-MSR at $k=0$, 
the agents in $\mathcal{G}_3$ will ignore every neighbor 
in $\mathcal{G}_1$ since $a < c$. 
Thus, for $i\in\mathcal{V}_3$, $\hat{x}_i[1]=a$. 
At $k=1$, the same happens for the agents $j \in \mathcal{V}_4$ 
and they stay at $b$. Since the agents in $\mathcal{G}_3$ 
are not affected by any agents with position values larger
than $a$, they remain at their positions for all $k \geq 0$. 
The same holds among the normal agents in the network, 
and therefore $\hat{x}_i[k]=a$ and $\hat{x}_j[k]=b$ for 
all $i \in \mathcal{V}_3$ and $j \in \mathcal{V}_4$. 
This shows failure in position agreement.
\end{proof}
}

\subsection{Further Results and Discussions}

Here, we provide some discussions related to the results of the paper
and their potential extensions. 

First, it is noteworthy that the result of Theorem~\ref{Theorem: DelaySuf} 
holds for the $f$-local malicious model as well, 
which is now stated as a corollary.
This follows since in the proof of the theorem, only the number 
of malicious agents in each normal agent's neighborhood plays a role. 

\begin{corollary}\label{Corollory: f-localAsycnronousSecond-Order}\rm
Under the $f$-local malicious model, the network of agents with 
second-order dynamics using the control in 
\eqref{eqn: DelayControl3} and the asynchronous DP-MSR algorithm 
reaches resilient consensus if the underlying graph is 
$(2f+1)$-robust. The safety interval is given 
by \eqref{eqn:  DelaySafetyInterval}. 
\end{corollary} 

{\color{black}
The underlying graph $\mathcal{G}$ in the results thus far has been fixed, 
that is, the edge set 
$\mathcal{E}$ is time invariant. This clearly limits application of
our results to, e.g., multi-vehicle systems. 
We now present an extension for a partially asynchronous time-varying network
$\mathcal{G}_0[k]=(\mathcal{V},\mathcal{E}[k])$, where the graph $\mathcal{G}_0[k]$
plays a role of the original graph $\mathcal{G}$ in the previous discussions. 
The following definition is the important for the development:




\begin{definition}\rm
The time-varying graph $\mathcal{G}_0[k]=(\mathcal{V},\mathcal{E}[k])$ is said to be
jointly $(2f+1)$-robust if there exists a fixed $h$ such that 
the union of $\mathcal{G}_0[k]$ over each consecutive $h$ steps is $(2f+1)$-robust.
\end{definition}

It is again assumed that time delays are upper
bounded by $\tau$ as in \eqref{eqn:delay_tau}
and moreover that $\tau$ is no less than the horizon parameter $h$ of $\mathcal{G}_0[k]$ as
\begin{equation}\label{As: assumption}
 h \leq \tau.
\end{equation} 
We now state the extension for time-varying networks. 

\begin{corollary}\label{Cor: Time-varying}\rm
Under the $f$-total/$f$-local malicious model, 
the time-varying network $\mathcal{G}_0[k]$ 
of agents with second-order dynamics using 
the control in \eqref{eqn: DelayControl3} and 
the asynchronous DP-MSR algorithm reaches resilient consensus
if $\mathcal{G}_0[k]$ is jointly $(2f+1)$-robust with the condition 
\eqref{As: assumption}. The safety interval is given 
by \eqref{eqn:  DelaySafetyInterval}.
\end{corollary}

The result follows from Theorem~\ref{Theorem: DelaySuf} since
in the proof there, the time-invariant nature of the original
graph $\mathcal{G}$ is not used. We also note that similar
development is made in \cite{LeBlancPaper} for the first-order 
agent networks; there, the assumption is that $\mathcal{G}_0[k]$ is 
$(f+1,f+1)$-robust at \textit{every time} $k$, which is obviously
more conservative than that in the above theorem.}

Next, we relate the graph properties that have appeared 
in the resilient consensus problem considered here
to those in standard consensus problems without attacks
\cite{CaoMorseAnderson,JiahuQin2012,JiahuHirche,Xiao2006}.
In this paper, we have assumed that the maximum number of 
faulty agents is at most $f$. In the case of $f=0$, 
all conditions 
in both synchronous and {\color{black}partially} asynchronous 
networks reduce to that of 1-robust graphs.
By Lemma~\ref{lemma:robust graphs}\;(iv), such a graph
is equivalent to having a directed spanning tree. 
It is well known that under such a graph, 
consensus can be achieved. 

{\color{black}
It is further noted that in the works \cite{Azadmanesh2002,khanafer,Azadmanesh1993,LeBlancPaper,plunkett,Vaidya}, malicious agents are allowed \textit{not} to make any transmissions, which is often called omissive faults. Hence, the normal agent $i$ would wait to receive at least $d_i-f$ values 
from its neighbors before making an update.
The necessary and sufficient condition derived in 
\cite{LeBlancPaper} on the network is $(2f+1,f+1)$-robustness. 
Compared to the synchronous case, an extra $f$ 
is needed because of the omissive faults, but the analysis 
remains mostly the same.
It should be noted that omissive faults can also be tolerated 
by the MSR-type algorithms of the paper. The malicious agents knowing that 
the normal agents apply the DP-MSR algorithm
might attempt to make this kind of attack to 
cause lack of information in step~2. In such cases, if agent $i$ does not receive the information of $m_i[k]$ incoming links at time $k$, then the parameter of the 
asynchronous DP-MSR for that agent can be changed 
from $2f$ to $2(f-m_i[k])$.


Furthermore, the DP-MSR algorithms for $f$-total malicious models 
are resilient against another type of attacks 
studied in \cite{Feng}. There, the attackers can create 
extra links in the networks, but adding links does 
not change the value of $f$ of the networks. 
In contrast, in the case with the $f$-local model, 
the situation is slightly different. Adding a link 
might increase the number of malicious neighbors
for some normal agents. Accordingly, 
the agents must be aware of the created links
so as to remove them along with
the edges ignored in step~2 of DP-MSR.
}

Finally, we discuss the two versions of consensus problems for
second-order agent networks.
In this paper, we have considered the case where
all normal agents agree on their positions and then stop. 
The other version is where normal agents aim at agreeing on both 
their positions and velocities, so they may be moving together 
at the end \cite{Cheng-Lin,renbook2,YU}.
This can be realized via the control law 
\begin{equation*}
  u_i[k]
   = \sum_{j \in \mathcal{N}_i} 
      a_{ij}[k] \bigl[
                 (\hat{x}_j[k]-\hat{x}_i[k])+\alpha(v_j[k]-v_i[k])
                \bigr].
\end{equation*}
This version of consensus is 
difficult to deal with based on MSR-type algorithms
when malicious agents are present. 
We may extend them so that not only the positions but also 
the velocities of extreme values are ignored at the time of
updates. It is however hard to prevent situations, for example, 
where
the normal agents follow a malicious agent
which moves away from 
them at a ``gentle'' speed within the normal range;
it seems unreasonable to call such a situation to be resilient.

\section{Numerical Example}\label{sect: simulations}


Consider the network of five agents in Fig.~\ref{fig:robustgraph}.
As mentioned earlier, this graph is $(2,2)$-robust. 
One of the agents is set to be malicious. 
The bound on the number of malicious agents is fixed at $f=1$.
We use the sampling period as $T=0.3$
and the parameter as $\alpha = 3.67$.


\begin{figure}[t]
	\centering
	\vspace*{-2mm}
	\includegraphics[width=78mm,height=40mm]{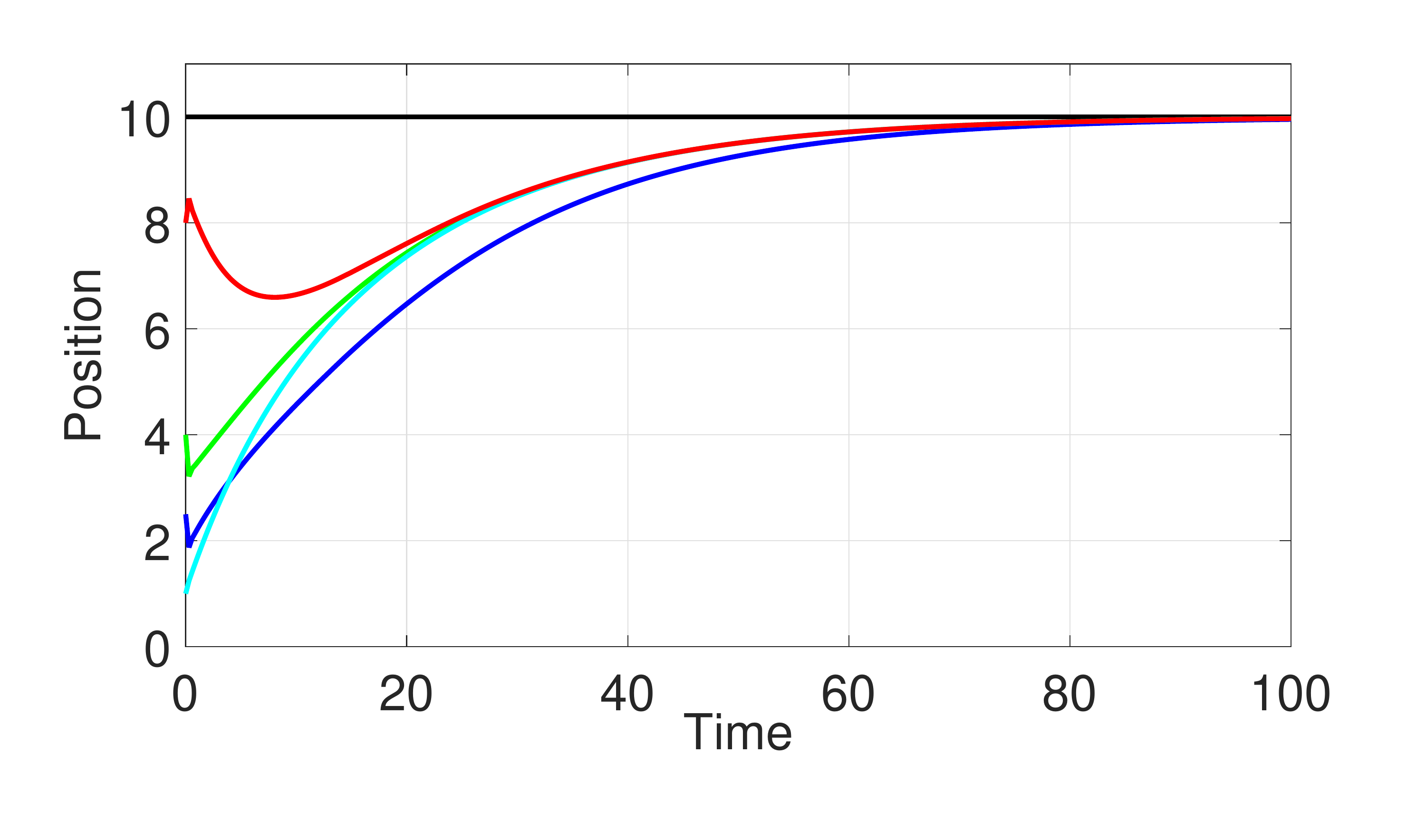}
	\vspace*{-7mm}
	\caption{Time responses for 
		the conventional algorithm over a (2,2)-robust graph.}
	\label{fig:syncposconventional}
\end{figure}

\textit{Synchronous Network}. First, we carried out simulations for the synchronous case. 
The initial states were chosen as
$\left[\hat{x}^T[0]~v^T[0]\right]
=\big[10~4~2.5~1~8~0~{-6}~{-5}~1~4\big]$. 
Agent~1 is set to be malicious and stays at its initial position 10. 
By \eqref{eqn:  SafetyInterval}, the safety interval is $\mathcal{S}=[0.19,8.54]$.
Fig.~\ref{fig:syncposconventional} shows the time responses of their
positions with the conventional control from \cite{renbook2} where 
in \eqref{eq3}, all $a_{ij}[k]$ are 
constant based on the original graph $\mathcal{G}$. 
The normal agents achieve agreement, but by following
the malicious agent. Thus, the safety condition is not met. 

 
Next, we applied the DP-MSR algorithm and the results are given in 
Fig.~\ref{fig:syncposrobust}. 
This time, the normal agents are not affected by the malicious agent~1 
and achieve consensus 
inside the safety region $\mathcal{S}$, confirming Theorem~\ref{NecSufSynch}.
In the third simulation, we modified the network in Fig.~\ref{fig:robustgraph}
by removing the edge $(2,5)$
and ran the DP-MSR algorithm.
The graph is in fact no longer $(2,2)$-robust. 
We see from Fig.~\ref{fig:syncposnotrobust} that
consensus is not attained. 
In fact, agent~5 cannot make any updates because it 
has only two neighbors after the removal of the edge. 

\begin{figure}[t]
	\centering
	\vspace*{-2mm}
	\includegraphics[width=78mm,height=40mm]{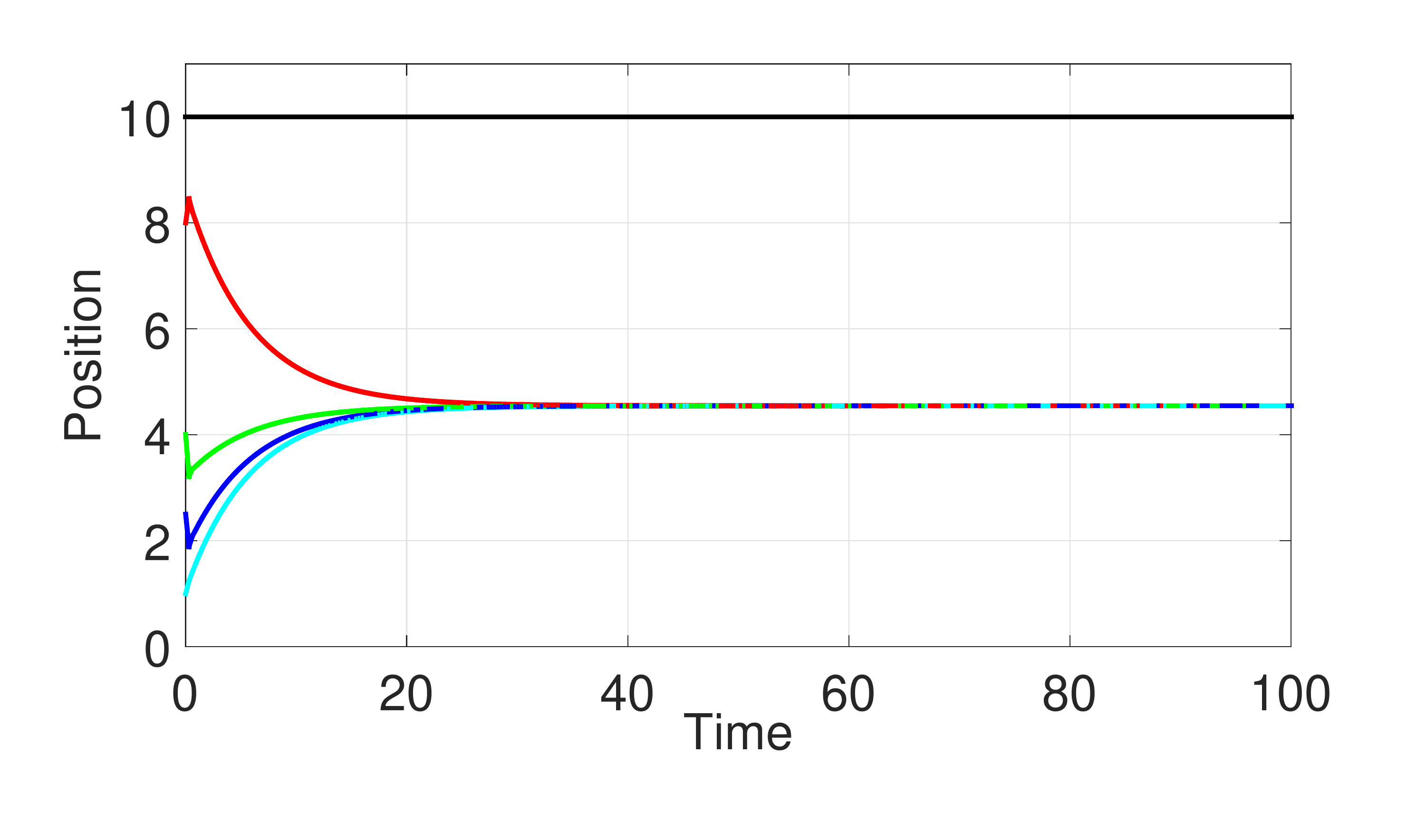}
	\vspace*{-7mm}
	\caption{Time responses for
		the synchronous DP-MSR algorithm over a (2,2)-robust graph.}
	\label{fig:syncposrobust}
\vspace*{1mm}
\includegraphics[width=78mm,height=40mm]{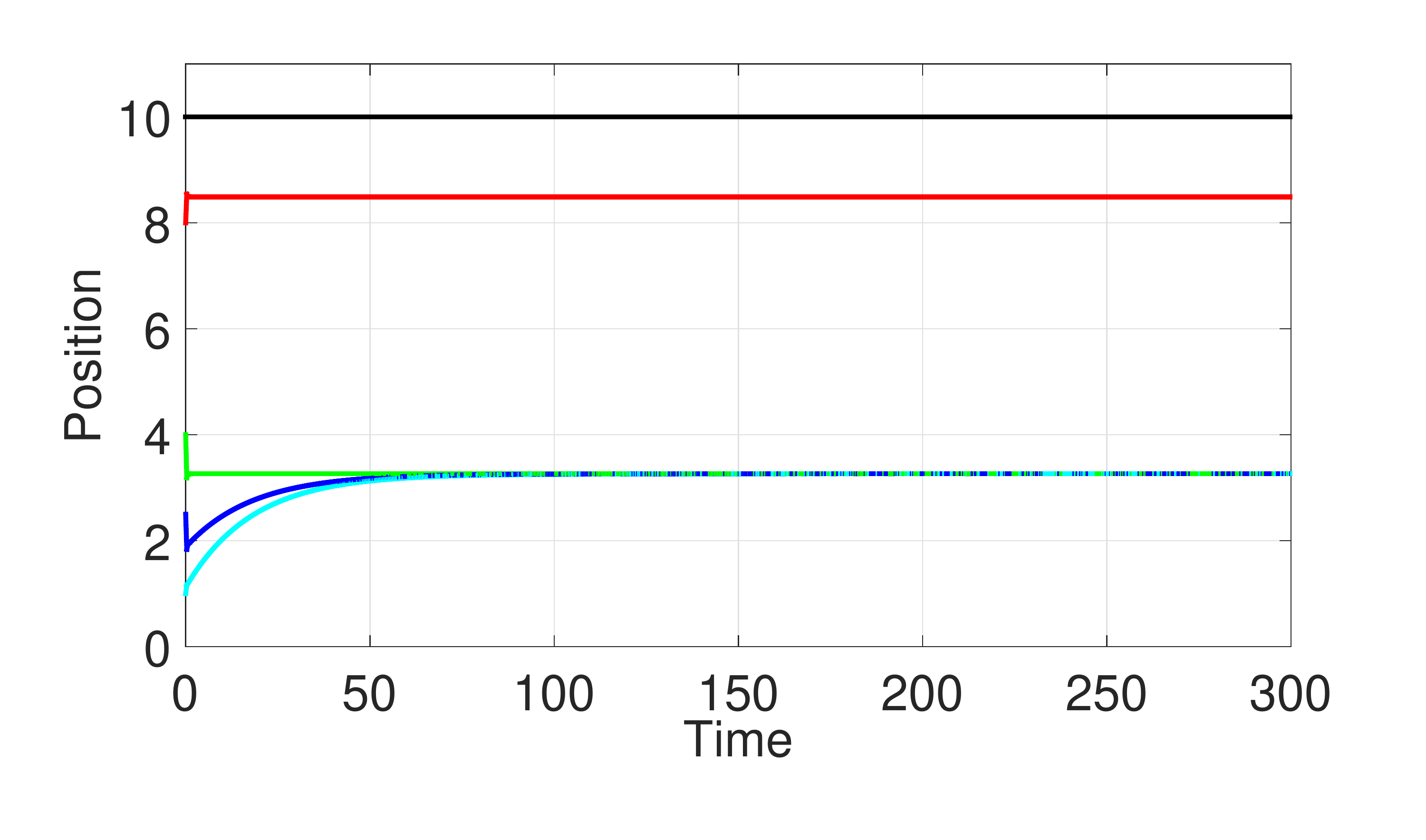}
\vspace*{-7mm}
\caption{Time responses for the synchronous DP-MSR algorithm over a non-robust graph.}
\label{fig:syncposnotrobust}
\end{figure}

\textit{Asynchronous Delayed Network}. Next, we examined the {\color{black}partially} asynchronous delayed version of the protocol. 
This time, agent 4 is chosen to be malicious.
Here, the normal agents make updates periodically
with period 12, but at different timings:
Agents 1, 2, 3, and 5 update at time steps 
$k=12\ell+6,12\ell+9,12\ell+11,12\ell+4$ for $\ell\in\mathbb{Z}_{+}$, respectively. 
We assume that at these time steps, there is no delay for their updates. 
Since the normal agents do not receive new data at other times, 
we have $\tau=11$. By this setting, each agent deals with nonuniform 
time-varying delays. 
The initial states of the agents are given 
by $\left[\hat{x}^T[0]~v^T[0]\right]
=\big[4~10~8~9~1~0~{-1}~{-1}~4~3\big]$. 
The safety interval in \eqref{eqn:  DelaySafetyInterval} 
becomes $\mathcal{S}_\tau=[0.865,10.54]$. In this case, the malicious 
agent~4 can move so that each normal 
agent sees it at different locations. 
We set its control as $u_4[k]=({(-1)^k}/{T^2})\left(-40T+ 14(2k+1)\right)$, 
which makes agent~4 oscillate as
$\hat{x}_4[2k]=2$ and $\hat{x}_4[2k+1]=9$, $k \geq 0$.

In Fig.~\ref{fig:asyncposnotrobust},
the time responses of the positions of normal agents are 
presented. 
Though the underlying network is $(2,2)$-robust, 
as the necessary condition 
in Theorem~\ref{Theorem: DelaySuf}, 
the normal agents do not come to consensus. 
This is an interesting situation since the asynchronous DP-MSR cannot prevent
the normal agents from being deceived by the malicious agent. Fig.~\ref{fig:asyncposnotrobust} indicates that in fact
agents~2 and 3 stay around 
$\hat{x}_i=9$ while agents~1 and 5 remain around $\hat{x}_i=3.7$.
So the agents are divided into two groups and settled at different
positions because of the malicious behavior of agent~4, appearing
at two different positions. 
%
Finally, we modified the graph 
to be complete,	
which is the only $3$-robust graph with 5 nodes. The responses in Fig.~\ref{fig:asyncposrobust} verify the sufficient condition of Theorem~\ref{Theorem: DelaySuf} for the {\color{black}partially} asynchronous setting.

\begin{figure}[t]
	\centering
	\vspace*{-2mm}
	\includegraphics[width=78mm,height=40mm]{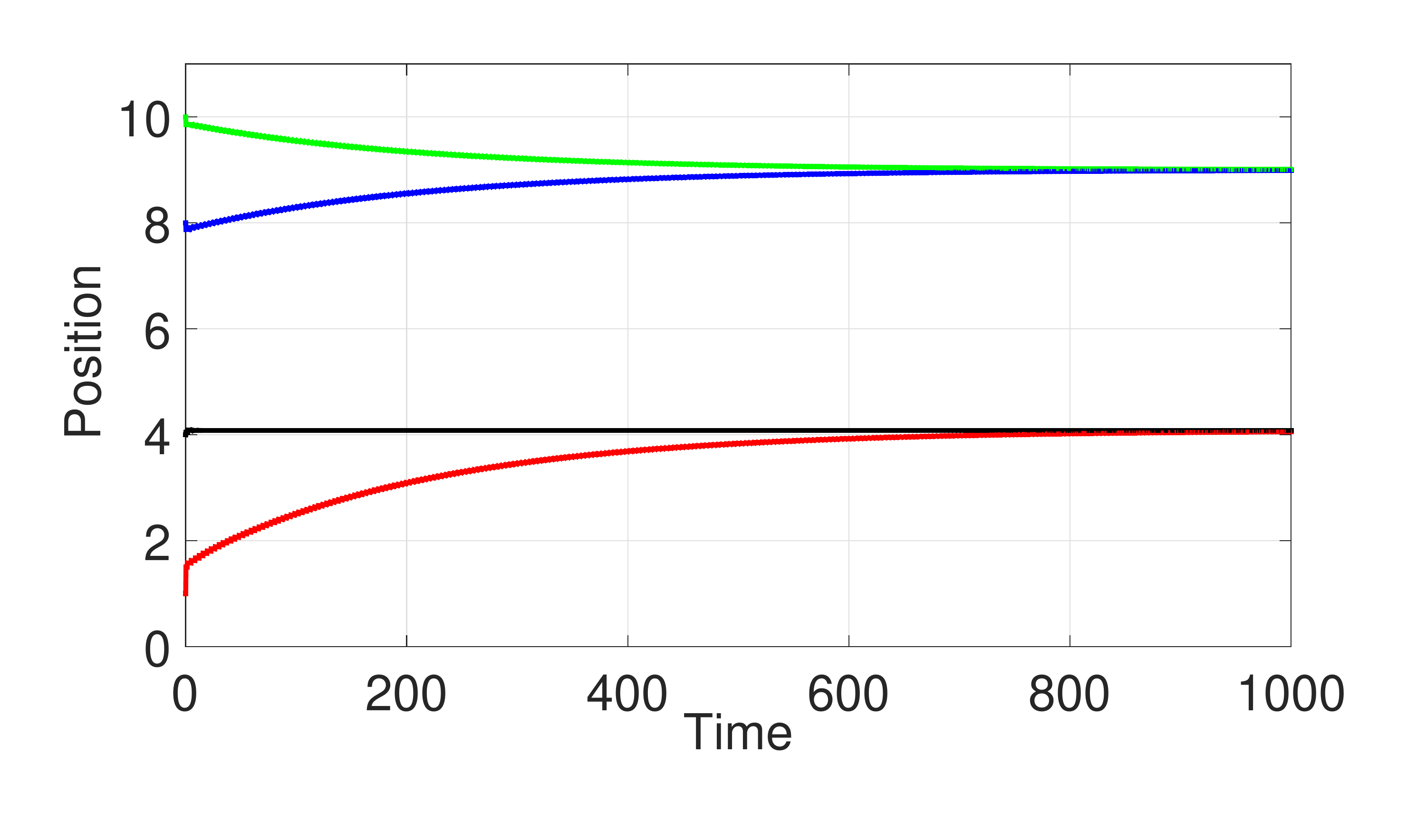}
	\vspace*{-7mm}
	\caption{Time responses for 
		the asynchronous DP-MSR algorithm over a (2,2)-robust graph.}
	\label{fig:asyncposnotrobust}
	\vspace*{1mm}
	\includegraphics[width=78mm,height=40mm]{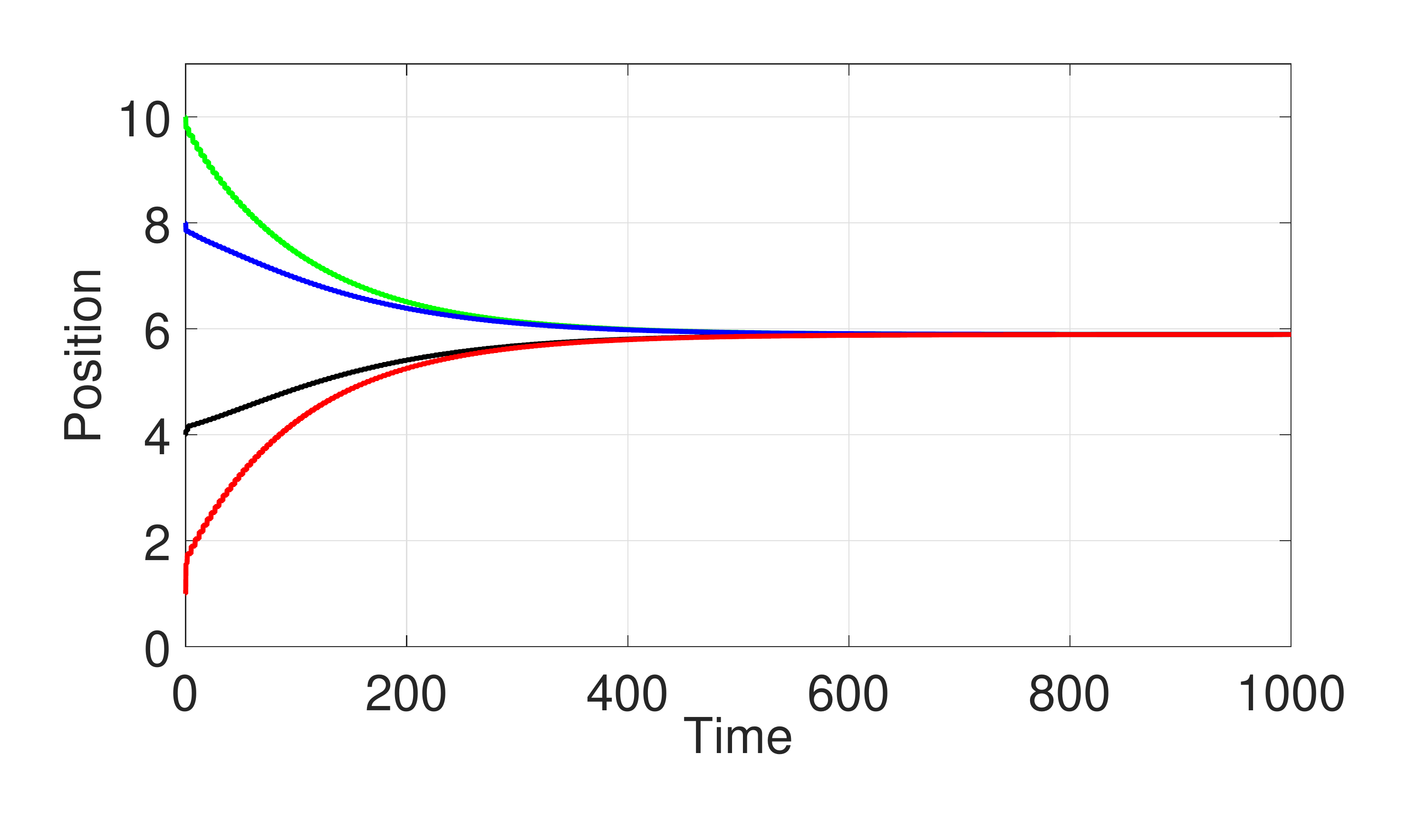}
	\vspace*{-7mm}
	\caption{Time responses for
		the asynchronous DP-MSR algorithm over a 3-robust graph.}
	\label{fig:asyncposrobust}
\end{figure}

\section{Conclusion}\label{sect: conclusion}

We have studied resilient consensus problems for
a network of agents with second-order dynamics, where
the maximum number of faulty agents in the network 
is known. 
We have proposed MSR-type algorithms for the normal agents 
to reach consensus
under both synchronous updates and {\color{black}partially} asynchronous 
updates with bounded delays.
Topological conditions in terms of robust graphs have been developed. 
Various comparisons have been made with related results 
in the literature of computer science.
%
In future research, we will consider using randomization in updates
to relax the robustness conditions (see \cite{dibajiishiitempoACC2016}).

\bibliographystyle{ifac}
\bibliography{main3}  

%
%

\end{document}